\documentclass[11pt,letterpaper]{article}
\usepackage[margin=1.05in,top=.91in,bottom=.93in]{geometry}
\usepackage{amsmath,amssymb,amsthm,booktabs,tabularx,array}
\usepackage{graphicx}
\usepackage{pgfplots}
\pgfplotsset{compat=1.17}
\usepackage{xcolor}
\usepackage[unicode,hidelinks,bookmarksnumbered]{hyperref}
\hypersetup{pdftitle={A Block Hierarchy for Rate-Distortion and Dispersion of Binary Markov Sources},pdfauthor={Bhaskar Krishnamachari}}
\definecolor{rateblue}{RGB}{39,86,135}
\definecolor{rateorange}{RGB}{181,85,32}
\definecolor{rateteal}{RGB}{24,122,111}
\pgfplotsset{blockaxis/.style={tick label style={font=\scriptsize},label style={font=\small},title style={font=\small},axis line style={black!55},tick style={black!55},axis lines*=left,ymajorgrids,grid style={black!12},legend style={draw=none,fill=none,font=\scriptsize,row sep=-1pt},scaled ticks=false}}

\newtheorem{theorem}{Theorem}
\newtheorem{proposition}[theorem]{Proposition}
\newtheorem{lemma}[theorem]{Lemma}
\newtheorem{corollary}[theorem]{Corollary}
\theoremstyle{definition}
\newtheorem{conjecture}{Conjecture}
\newcommand{\Vhier}{V_{\mathrm{hier}}}
\newcommand{\E}{\mathbb E}
\newcommand{\Prb}{\mathbb P}
\newcommand{\Var}{\operatorname{Var}}
\newcommand{\Cov}{\operatorname{Cov}}
\newcommand{\sd}{\operatorname{sd}}
\newcommand{\ind}{\mathbf 1}
\newcommand{\jj}{\jmath}

\newcommand{\Vfun}{\mathcal V}
\newcommand{\hx}{h_2}
\title{\bfseries A Block Hierarchy for Rate--Distortion and\\Dispersion of Binary Markov Sources}
\author{Bhaskar Krishnamachari\\[3pt]
\small Ming Hsieh Department of Electrical and Computer Engineering\\[-2pt]
\small University of Southern California, Los Angeles, CA\\[-2pt]
\small\texttt{bkrishna@usc.edu}}
\date{}
\begin{document}
\maketitle

\begin{abstract}
The finite-length rate penalty for compressing a binary Markov source under Hamming distortion remains open in general. To study this penalty, we use memoryless sources that draw independent $K$-symbol blocks from the Markov law. We determine the rate and tilted-information variance per symbol at every distortion for positively correlated two-symbol blocks of arbitrary stationary bias and for symmetric three- and four-symbol blocks. For a symmetric chain with flip probability $0<p<1/2$, we prove that the normalized variance converges to a positive limit at every interior distortion where the Markov rate--distortion function is differentiable. At positive distortions up to Gray's critical distortion, this limit equals the source varentropy rate and yields a second-order coding converse for a single Markov path. We conjecture that the limit is the operational dispersion of the Markov source. Matching achievability remains open.
\end{abstract}

\section{Introduction}
A symmetric binary Markov source can spend long stretches in one state even though each bit is equally likely to be zero or one. If its flip probability is $p<1/2$, two successive bits distinguish a rare flip from a common repetition. How much additional rate does a finite code need, beyond the asymptotic compression rate, to meet a specified probability of excess Hamming distortion?

The rate--distortion function $R(D)$ gives the asymptotic number of bits per source symbol needed to reproduce a source with average distortion at most $D$. For a memoryless source with symbol $X$ and distortion measure $d$, Shannon's characterization is
\[
 R(D)=\min_{P_{\hat X\mid X}:\,\E d(X,\hat X)\le D} I(X;\hat X).
\]
The minimization chooses a conditional reproduction law that meets the distortion constraint while retaining as little information about $X$ as possible. For binary Hamming distortion, $d(x,\hat x)$ is one for a mismatch and zero for a match, so $D$ is an allowed fraction of errors. For a stationary Markov source, the minimization instead uses the joint law of $K$ consecutive symbols and a reproduction block of the same length. With $R_K(KD)$ denoting the minimum mutual information at mean Hamming distance at most $KD$, the classical characterization is \cite{berger71,gray71,wynerziv}
\[
       R(D)=\lim_{K\to\infty}r_K(D),\qquad r_K(D)=\frac{R_K(KD)}{K}.
\]

At finite length the constraint is probabilistic. A code for $n$ symbols selects one of $M$ reproduction words. Let $R^*(n,D,\epsilon)$ be the smallest rate $(\log_2 M)/n$ for which the fraction of errors exceeds $D$ with probability at most $\epsilon$. For a discrete memoryless source, Ingber and Kochman \cite{ik} and Kostina and Verd\'u \cite{kv} give the normal approximation
\begin{equation}
 R^*(n,D,\epsilon)
 =R(D)+\sqrt{\frac{V(D)}{n}}\,Q^{-1}(\epsilon)+o(n^{-1/2}).
 \label{eq:iidnormal}
\end{equation}
Here $0<\epsilon<1$ is fixed. We use the theorem for bounded distortion at an interior $D$ where $R$ is differentiable and the optimal test channel is unique, under the regularity conditions of \cite{kv}. The function $Q(z)$ is the upper-tail probability of a standard normal random variable. The quantity $V(D)$ is the \emph{dispersion}, measured in bits squared per symbol. For example, with $\epsilon=0.1$ the leading rate penalty is $1.282\sqrt{V(D)/n}$ bits per symbol. Thus two sources with the same $R(D)$ can require different rates at the same finite length.

Kostina and Verd\'u identify the random quantity whose variance gives $V(D)$. Let $q^*$ be an optimal reproduction marginal and $\lambda=-R'(D)$ at a differentiable interior point. The \emph{$d$-tilted information} of a source symbol $x$ is
\[
 \jmath(x,D)=-\log_2\sum_y q^*(y)\,2^{\lambda(D-d(x,y))}.
\]
The sum assigns the greatest weight to reproduction symbols near $x$. If the optimal reproduction law assigns little mass near $x$, the sum is small and $\jmath(x,D)$ is large. The tilted information has mean $R(D)$ and variance $V(D)$. For independent symbols, tilted information adds. Equation~\eqref{eq:iidnormal} says that, to second order, the best rate is the $(1-\epsilon)$-quantile of this sum divided by $n$ \cite{kv}.

For the symmetric chain in the opening example, the single-letter tilted information is constant. At sufficiently low distortion and $0<p<1/2$, a flip carries $\log_2((1-p)/p)$ more bits of tilted information than a repetition.

\begin{figure}[!b]\centering
\includegraphics[width=.92\linewidth]{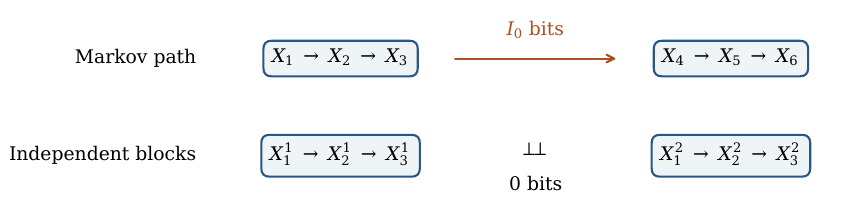}
\caption{Two consecutive blocks of length three. Within either block the symbols follow the transition matrix $P$. Along the path the mutual information between the blocks equals $I(X_3;X_4)=I_0$ by the Markov property. Independent draws from the same three-symbol marginal remove exactly this boundary dependence.}\label{fig:construction}
\end{figure}

\paragraph{The block construction.}
Consider a single uninterrupted path of the Markov source. Its consecutive $K$-symbol blocks share dependence through the transition at each boundary. Now draw blocks independently, preserving the Markov law within each block (Figure~\ref{fig:construction}). The resulting source is memoryless on an alphabet of $2^K$ binary words, so its tilted-information variance $V_K(KD)$ gives its dispersion under the hypotheses of the memoryless theorem. The \emph{block hierarchy} consists of these sources indexed by $K$ at the same per-symbol distortion $D$. Each boundary accounts for $I_0=I(X_1;X_2)$ bits of dependence, and $0\le r_K(D)-R(D)\le I_0/K$. Does the normalized variance $v_K(D)=V_K(KD)/K$ converge as the blocks grow? If so, does its limit govern the finite-length rate of an uninterrupted Markov path?

\begin{figure}[!b]\centering
\includegraphics[width=\linewidth]{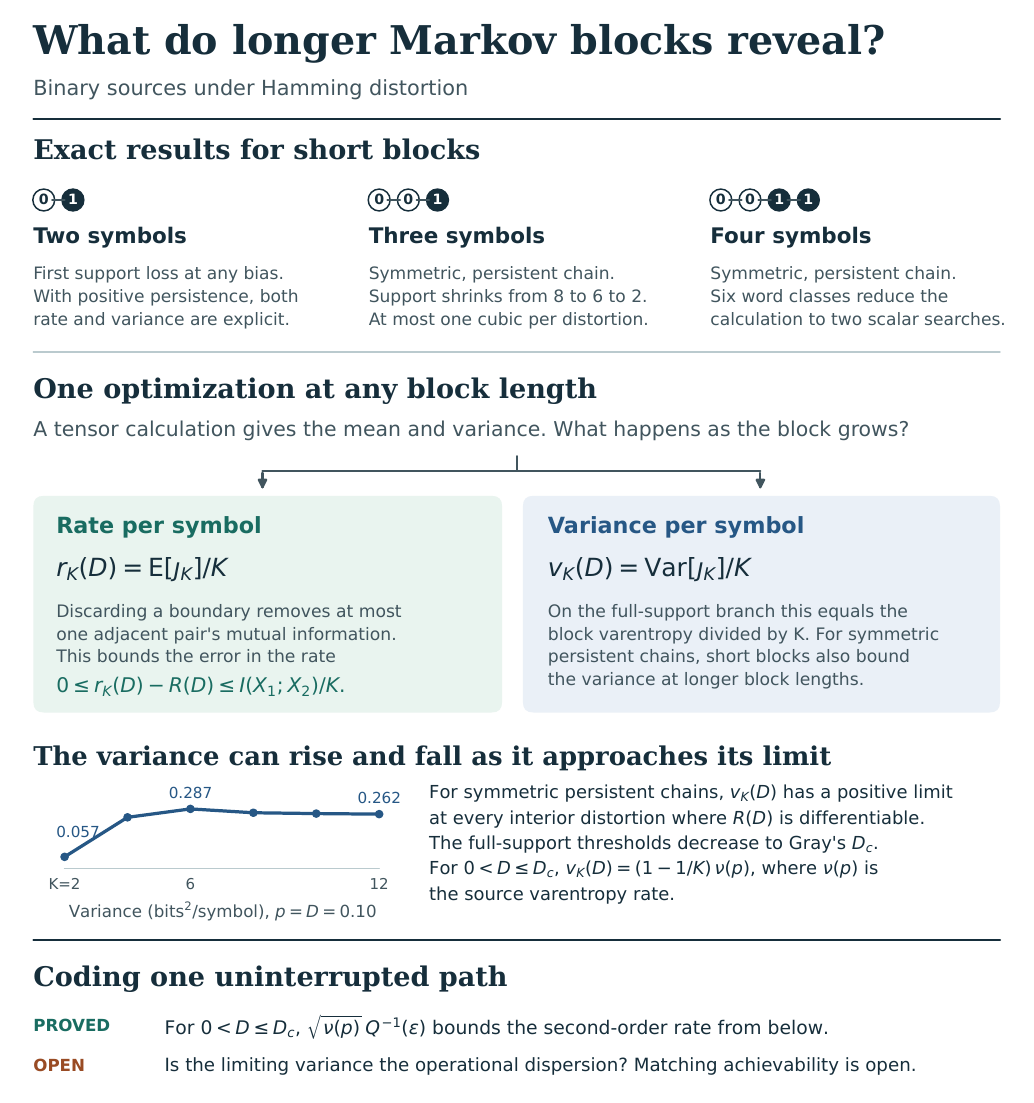}
\caption{Results of this paper at a glance.}\label{fig:overview}
\end{figure}

\paragraph{Results.}
Figure~\ref{fig:overview} traces the two quantities obtained from each block, the rate and the tilted-information variance, from short words to the limit of long blocks. For positively correlated two-symbol blocks we compute the complete rate and variance at any stationary bias. For symmetric three- and four-symbol blocks we identify the support changes and compute both quantities with at most a cubic equation or two scalar searches, respectively (Section~\ref{sec:two}).

For larger blocks, one optimization gives both moments, and shorter blocks bound the rate and, for symmetric persistent sources, the variance (Sections~\ref{sec:general} and~\ref{sec:shortvariance}). The finite-block calculation depends on which words receive positive probability under the optimal reproduction law. These words form its \emph{support}. A \emph{branch} is an interval of distortion with the same support. On the \emph{full-support branch}, all reproduction words have positive probability. For symmetric persistent chains, the full-support thresholds decrease with block length to Gray's critical distortion $D_c$. Up to $D_c$ the entropy formula for the block rate holds at every length. Above $D_c$, longer blocks can lose support at a fixed distortion.

Figure~\ref{fig:regimes} follows one symmetric source as the block length grows. At $D=0.05$, every block length lies below its full-support threshold, so the variance curve is exact. Theorem~\ref{thm:common} also establishes convergence beyond that common interval wherever $R$ is differentiable in the interior. Proposition~\ref{prop:converse} relates the low-distortion limit to coding a single path.

\begin{figure}[!htbp]\centering
\includegraphics[width=\linewidth]{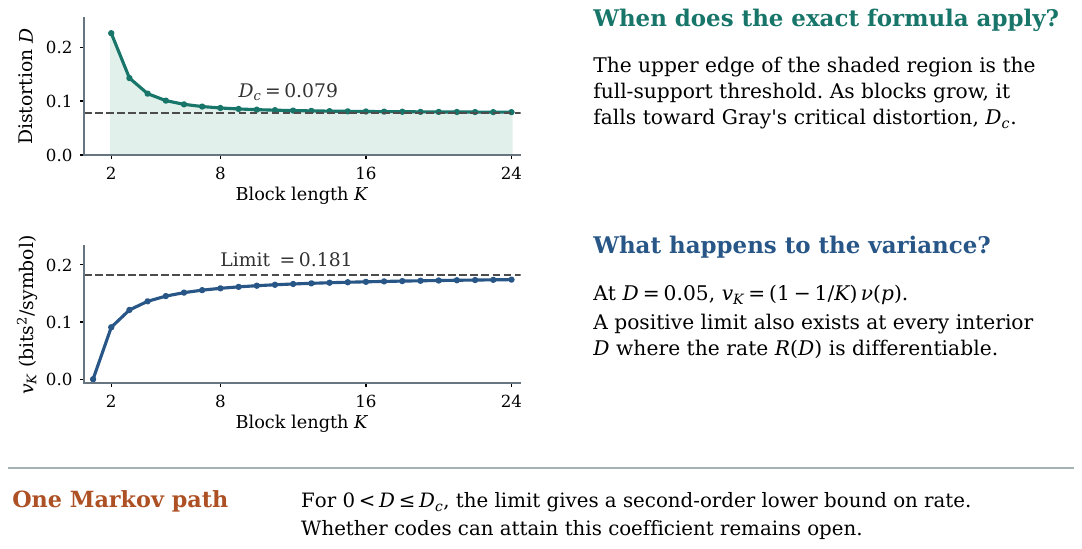}
\caption{From finite blocks to one Markov path for a symmetric source with flip probability $p=0.35$. Points represent integer block lengths. The dashed lines give the limiting full-support threshold $D_c$ (top) and the source varentropy rate $\nu(p)$ (bottom).}\label{fig:regimes}
\end{figure}

\paragraph{Code, data, and proofs.}
The plotting code for all thirteen figures, the finite-block solver, and the numerical data are available at \url{https://github.com/ANRGUSC/markov-block-hierarchy}. The \href{https://github.com/ANRGUSC/markov-block-hierarchy/tree/main/proofs}{Lamport proof companion} gives step-by-step proofs of all ten numbered results and identifies the dependencies used to verify each one.

\paragraph{Relation to prior work.}
This paper extends the workshop paper and presentation \cite{itw}. They introduced the hierarchy, identified the first two-symbol support change as a function of bias and persistence, and examined varentropy and finite-block numerical behavior.

Earlier versions of this preprint studied single-letter tilted information summed along a Markov path.\footnote{For a biased chain, the centered sum equals the centered count of ones times $\log_2(\pi_0/\pi_1)$, where $\pi_x=\Prb(X_1=x)$.} For a symmetric chain this sum is constant. The present revision studies the fluctuations of successively longer blocks.

Classical block bounds and Gray's low-distortion results give the first-order rate for sources with memory \cite{gray70,gray71,wynerziv,berger77,jalali}. Blahut computed binary Markov block rates through length ten \cite{blahut}. For each fixed block length, the independent-block construction also brings memoryless finite-blocklength theory to bear \cite{ik,kv,zhoumotani}. This theory draws on the normal-approximation tradition in information theory \cite{strassen,ppv}. In particular, Kostina \cite{kostina17} relates a tight Shannon lower bound to dispersion equal to source varentropy, the variance of the source self-information. This accounts for the variance formula on our full-support branch.

For a single Markov path, lossless coding theory connects varentropy to Gaussian approximation bounds \cite{kontoverdu}. Related second-order results cover Markov source coding with decoder side information and channels with Markov noise \cite{hayashi}. In lossy coding, the dispersion is known for Gauss--Markov sources under squared error and for more general Gaussian sources \cite{tian,tasci}. The corresponding identification for a binary Markov path under Hamming distortion remains open in general.

The achievability question concerns the probability that a random reproduction word falls within the distortion budget of a source path. Kontoyiannis \cite{konto00} studied pointwise redundancy for variable-length lossy coding using a fluctuation term that, after centering, agrees with tilted information in the memoryless setting. Kontoyiannis and Zhang \cite{kz02} gave direct and converse bounds for arbitrary sources in terms of reproduction probability measures. Section~\ref{sec:conclusion} describes the distortion-ball estimate needed here.

\section{Block rate, tilted information, and dispersion}\label{sec:setup}
Let $X_1,X_2,\ldots$ be a stationary binary Markov chain with transition matrix $P$, stationary law $\pi$, and persistence $\rho$ given by
\[
 P=\begin{pmatrix}1-a&a\\ b&1-b\end{pmatrix},\qquad
 \pi_0=\frac b{a+b},\quad \pi_1=\frac a{a+b},\quad \rho=1-a-b .
\]
We assume $0<a,b<1$ unless otherwise specified. Thus $\pi_1$ determines the marginal bias and $\rho$ is the lag-one correlation. Positive persistence means $\rho>0$, negative persistence means $\rho<0$, and $\rho=0$ gives an i.i.d.\ source. For a binary word $x$ of length $K$, let $P_K(x)=\pi_{x_1}\prod_{i=2}^K P_{x_{i-1}x_i}$. Write $d(x,y)=\sum_{i=1}^K\ind\{x_i\ne y_i\}$ for Hamming distance and $\hx$ for binary entropy in bits. We use \emph{block} for $K$ source positions and \emph{word} for a particular binary string filling those positions. A per-symbol distortion $D$ gives a block budget of $KD$.

\subsection{Two moments of the block tilted information}
For $X\sim P_K$ and a conditional reproduction law $W(y\mid x)$, define
\begin{equation}
 R_K(KD)=\min_{W:\,\E d(X,Y)\le KD} I(X;Y),\qquad
 r_K(D)=R_K(KD)/K .                                           \label{eq:rd}
\end{equation}
At $D=0$, $R_K(0)=H(X^K)$. At $D_0=\min(\pi_0,\pi_1)$, the more frequent symbol repeated $K$ times attains rate zero. Consider an interior distortion $0<D<D_0$. At a differentiability point, $r_K(D)=R_K(KD)/K$ gives $r_K'(D)=R_K'(KD)$, where the derivative on the right is with respect to the block-distortion budget. Thus the slope magnitude $\lambda=-r_K'(D)$ has units of bits per error. At a corner choose a supporting slope and denote its magnitude by $\lambda$. Let $q$ be an optimal reproduction marginal. We use three equivalent parameters for this slope,
\[
 \underbrace{\lambda=-R_K'(KD)}_{\text{bits per error}},\qquad
 \underbrace{\tau=\lambda\ln 2}_{\text{nats per error}},\qquad
 \underbrace{u=2^{-\lambda}=e^{-\tau}}_{\text{Hamming kernel parameter}}.
\] For $0<u<1$ the optimal channel and its normalizer are
\begin{equation}
 Z(x)=\sum_y q(y)u^{d(x,y)},\qquad W(y\mid x)=\frac{q(y)u^{d(x,y)}}{Z(x)}.
 \label{eq:Z}
\end{equation}
Applying the single-symbol definition from the introduction to one $K$-symbol word gives
\begin{equation}
 \begin{aligned}\jj_K(x,D)&=KD\log_2u-\log_2 Z(x),\\
 \E\jj_K(X,D)&=R_K(KD),\qquad V_K(KD)=\Var[\jj_K(X,D)].\end{aligned}        \label{eq:j}
\end{equation}
Substitution of \eqref{eq:Z} in the mutual information gives the mean identity. The normalized variance $v_K(D)=V_K(KD)/K$ is measured in bits$^2$ per source symbol. At $D=0$ we take the lossless tilted information to be $-\log_2 P_K(x)$. At and above $D_0$ we choose the supporting slope $\lambda=0$. Then $u=1$, $Z(x)=1$, and the tilted information and its variance are both zero. At a biased endpoint the limit from below may differ from this assigned value.

\subsection{Independent blocks and one Markov path}
Consider first a source that draws $m$ independent words from $P_K$. At fixed $K$, its letters are memoryless and each letter contains $K$ Markov symbols. Write $R^*_{m,K}(D,\epsilon)$ for the minimum $\log_2 M/(mK)$ over codes of size $M$ with probability at most $\epsilon$ of exceeding distortion $mKD$. Under the interior-point assumptions of the introduction for this finite block alphabet, the memoryless normal approximation gives \cite{kv}
\begin{equation}
 R^*_{m,K}(D,\epsilon)=r_K(D)+\sqrt{\frac{v_K(D)}{mK}}\,Q^{-1}(\epsilon)
          +O_K\!\left(\frac{\log m}{mK}\right).                  \label{eq:oper}
\end{equation}
The sharper remainder follows from the finite-alphabet memoryless bounds of \cite{kv} for $m$ independent blocks, after division by $K$ to express the rate per symbol. Equation~\eqref{eq:iidnormal} only retains terms through order $n^{-1/2}$. Here $K$ is fixed and the constant in the bound may depend on $K$.

The distinction is operational: \eqref{eq:oper} describes independent draws from $P_K$. Consecutive blocks of the Markov source itself remain dependent across every boundary, even when their common marginal is exactly $P_K$.
\subsection{Optimal reproduction support}\label{sec:full}
To determine the reproduction support at fixed $u$, we use the following variational characterization. At this multiplier, the optimal $q$ maximizes $\sum_x P_K(x)\ln Z(x)$ on the probability simplex \cite{blahut}. This follows by minimizing the mean relative entropy from $W(\cdot\mid X)$ to $q$, plus the distortion penalty, over $W$ and $q$ in turn. The optimal $q$ is unique because $P_K$ is positive, the logarithm is strictly concave, and the Hamming matrix $[u^{d(x,y)}]$ is invertible. The necessary and sufficient optimality conditions are
\begin{equation}
 c(y)=\sum_x P_K(x)\frac{u^{d(x,y)}}{Z(x)}\le1,\qquad
 c(y)=1\ \text{whenever }q(y)>0.                                \label{eq:kkt}
\end{equation}
We call $c(y)$ the \emph{column value} of reproduction word $y$, indexing a column of the Hamming matrix $[u^{d(x,y)}]$. Its upper bound in \eqref{eq:kkt} is the column test. It is the derivative of $\sum_xP_K(x)\ln Z(x)$ with respect to $q(y)$ and $\sum_y q(y)c(y)=1$.

At sufficiently low distortion all reproduction words are used. In that range an optimal backward channel generates $X$ from $Y$ by flipping each bit independently with probability $D$. To see why, let $E=X\oplus Y$. The entropy bound $H(X\mid Y)=H(E\mid Y)\le\sum_i H(E_i)\le K\hx(D)$ is tight precisely when the error bits are independent Bernoulli$(D)$ and independent of $Y$. At $u=D/(1-D)$ the inverse of this backward channel gives the candidate reproduction law
\begin{equation}
 q_D(y)=\frac{1}{(1-u)^K}\sum_x(-u)^{d(x,y)}P_K(x).              \label{eq:inverse}
\end{equation}
Its entries sum to one. Whenever they are nonnegative, the backward-channel slope $\lambda=\log_2((1-D)/D)$ gives $Z(x)=(1+u)^K P_K(x)$ and hence
\begin{equation}
 \begin{aligned}
 \jj_K(x,D)&=-\log_2P_K(x)-K\hx(D),\\
 R_K(KD)&=H(X^K)-K\hx(D),\qquad V_K(KD)=\Var[-\log_2P_K(X)] .
 \end{aligned}                                                   \label{eq:slb}
\end{equation}
The variance of the self-information $-\log_2 P_K(X)$ is the \emph{block varentropy}. These identities extend to the first zero of $q_D$ at the same backward-channel slope. If that zero occurs at a biased zero-rate endpoint $D_0$, Section~\ref{sec:setup} assigns $v_K(D_0)=0$ by choosing $\lambda=0$. For independent Bernoulli$(1/4)$ symbols at $D_0=1/4$, the tilted-information variance per symbol is $(3/16)\log_2^2 3$ at the backward-channel slope and zero at the assigned slope. Conversely, if all optimal masses are positive, the equalities in \eqref{eq:kkt} imply $P_K(x)/Z(x)=(1+u)^{-K}$ by inversion of the same matrix. The backward channel therefore has crossover $u/(1+u)=D$. This is the tight Shannon lower bound for binary Hamming distortion \cite{berger71,gray70}. The resulting equality between dispersion and varentropy is a special case of \cite{kostina17}; the workshop paper \cite{itw} specializes it to the binary Markov block law. Its feasibility depends on the source, even though the backward crossover has the same expression for every source.

\section{Explicit results for two, three, and four symbols}\label{sec:two}
For fixed $K$, the optimal reproduction law determines both the rate and the variance through the same normalizer $Z(x)$. The main calculation is therefore to identify the reproduction words that retain positive mass as distortion increases. For two symbols this calculation exposes the roles of bias and persistence. For three and four symbols, symmetry reduces it to a small number of word classes. Figure~\ref{fig:small} compares the resulting rate and variance curves.
Table~\ref{tab:roadmap} summarizes what can be computed at each length.
\begin{table}[htbp]\centering\small
\caption{Computing the block rate and variance at distortion $D$ per symbol.}
\label{tab:roadmap}\begin{tabularx}{\linewidth}{@{}l l X@{}}\toprule
Length & Source family & Thresholds, rate and variance\\\midrule
2 & Any bias & First support loss in closed form for either sign of persistence.\\
2 & Positive persistence & Complete rate and variance in elementary functions.\\
3 & Symmetric & Two explicit thresholds. At most one cubic per distortion.\\
4 & Symmetric & Six classes. Thresholds and both moments from two scalar searches.\\
Any $K$ & $0<a,b<1$ & Positive transition probabilities allow either sign of persistence. Tensor optimization gives both moments, with exact formulas on the full-support branch. Short blocks bound the rate and, for symmetric persistent chains, the variance.\\\bottomrule
\end{tabularx}
\end{table}
\subsection{Two symbols: bias and persistence}
Two symbols are enough to distinguish persistence from bias. Let $A=P_2(00)$, $B=P_2(11)$ and $C=P_2(01)=P_2(10)$. Stationarity gives
\begin{equation}
 C=(1-\rho)\pi_0\pi_1,\quad A=\pi_0-C,\quad B=\pi_1-C,\qquad
 A-B=\pi_0-\pi_1,\quad AB-C^2=\rho\pi_0\pi_1 .                 \label{eq:pairmass}
\end{equation}
The two differences in \eqref{eq:pairmass} determine which word first receives zero reproduction probability as $D$ increases. Denote the full-support threshold at block length $K$ by $D_{\rm FS}(K)$, abbreviated $D_{\rm FS}$ here for $K=2$. Figure~\ref{fig:phase} shows both the identity of the first lost word and the distortion at which it is lost.

\begin{figure}[htbp]\centering
\includegraphics[width=\linewidth]{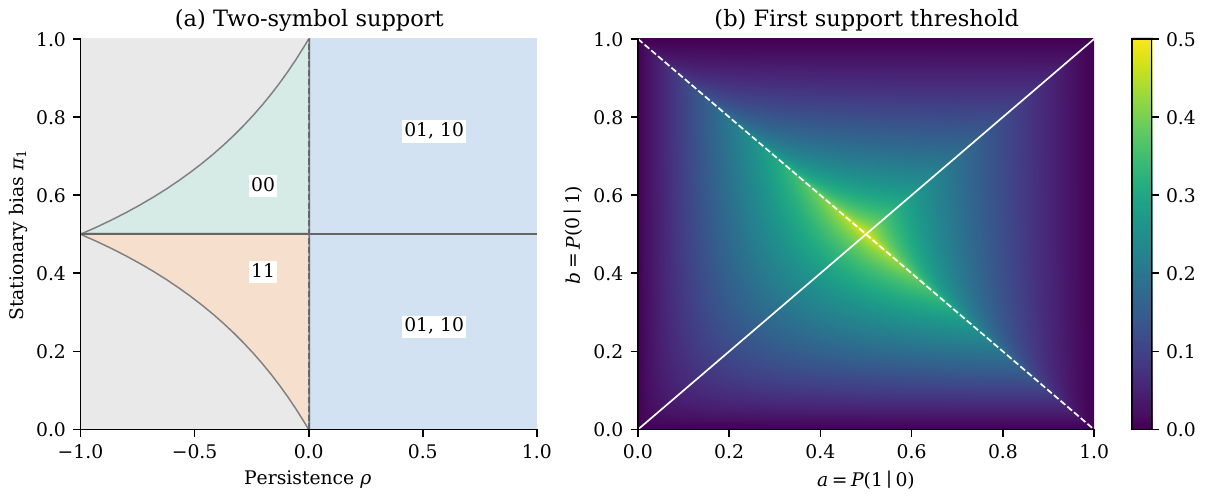}
\caption{First two-symbol support threshold for biased chains, redrawn from \cite{itw}. The initial support is $\{00,01,10,11\}$ throughout. In (a) each label names the reproduction words lost first. At the zero-rate threshold $D_0$, the constant reproduction is $00$ below the horizontal unbiased line and $11$ above it. Gray denotes infeasible parameters. The dashed vertical line marks independence. In (b) color gives the threshold distortion over the transition probabilities. The solid and dashed diagonals mark unbiased and independent sources, respectively.}
\label{fig:phase}
\end{figure}

\begin{proposition}[First support change \cite{itw}]\label{prop:phase}
If $\rho>0$, both $01$ and $10$ leave at $D_{\rm FS}=(1-\sqrt{1-4C})/2$. If $\rho<0$, the rarer constant leaves first, with both constants leaving together at $\pi_0=\pi_1$. In this case
\[
 u_*=\frac{C-\sqrt{C^2-AB}}{\max(A,B)},\qquad D_{\rm FS}=\frac{u_*}{1+u_*}.
\]
If $\rho=0$, full support persists for all $D<D_0$.
\end{proposition}
\begin{proof}
The full-support probabilities \eqref{eq:inverse} at $u=D/(1-D)$ are
\begin{align}
 q_{00}&=\frac{A-2Cu+Bu^2}{(1-u)^2},\qquad q_{11}=\frac{B-2Cu+Au^2}{(1-u)^2},\nonumber\\
 q_{01}=q_{10}&=\frac{C(1+u^2)-(A+B)u}{(1-u)^2}.               \label{eq:pairprior}
\end{align}
For $\rho>0$, the constant-word quadratics have negative discriminant. The transition-word numerator vanishes at $D(1-D)=C$. For $\rho<0$, the smaller constant mass gives the root $u_*$, where the transition-word numerator is $(1-u_*^2)\sqrt{C^2-AB}>0$. The constant-word numerators coincide at zero bias. For $\rho=0$, $q$ is a product of Bernoulli laws with parameter $(\pi_1-D)/(1-2D)$.
\end{proof}

For example, $(a,b)=(0.2,0.4)$ and $(0.4,0.8)$ have the same stationary bias $\pi_1=1/3$. In the persistent chain, $01$ and $10$ disappear together at $D=0.15843$. In the anti-persistent chain, $11$ disappears first, at $D=0.12251$, leaving $\{00,01,10\}$. Exchanging $a$ and $b$ exchanges the identity of the lost constant. On the unbiased anti-persistent line the two transition words remain equiprobable all the way to zero rate, as follows by flipping the second bit of both input and reproduction.

A biased anti-persistent chain can lose further words before reaching zero rate. For example, take $(a,b)=(3/4,7/8)$, so $(A,C,C,B)=(7,21,21,3)/52$. The first word to leave is $11$. At $u=1/2$, equal masses on $01$ and $10$ give normalizers $1/2$ on the constants and $5/8$ on the transition words. The active column values are one, while $c(00)=491/520$ and $c(11)=431/520$. Thus \eqref{eq:kkt} certifies support $\{01,10\}$ at $D=67/260$. By the zero-rate threshold $D_0=6/13$, only $00$ attains zero rate at the prescribed distortion.

\subsubsection{Complete rate and variance for positive persistence}
Positive persistence gives a particularly simple second branch. Once $01$ and $10$ leave the support, either constant reproduction incurs exactly one error on an input containing a flip. The distortion available to distinguish the two constant inputs determines the rest of the curve.

\begin{theorem}[Persistent two-symbol source]\label{thm:two}
Suppose $\rho>0$ and write $S=A+B$. The rate in bits per pair is
\begin{equation}
 R_2(2D)=
 \begin{cases}
 H(A,C,C,B)-2\hx(D), &0\le D\le D_{\rm FS},\\
 S[\hx(B/S)-\hx((D-C)/S)], &D_{\rm FS}<D<D_0,\\
 0, &D\ge D_0 .
 \end{cases}                                                     \label{eq:r2}
\end{equation}
On the first branch the variance is the varentropy of $(A,C,C,B)$. On the second branch put $\delta_B=(D-C)/S$ and $u=\sqrt{\delta_B/(1-\delta_B)}$. The variance $V_2(2D)$ is the variance of $-\log_2Z$ under the three class probabilities
\begin{equation}
 (s_{00},s_{\rm tr},s_{11})=(A,2C,B),\qquad
 (Z_{00},Z_{\rm tr},Z_{11})=\left(\frac{A(1+u^2)}{S},u,\frac{B(1+u^2)}{S}\right).
 \label{eq:v2b}
\end{equation}
The constant reproduction words remain active for $D_{\rm FS}\le D<D_0$. For a biased source, the more frequent constant is the only zero-rate reproduction at $D_0$. At $D\ge D_0$ the zero supporting slope gives $V_2(2D)=0$.
\end{theorem}
\begin{proof}
Equation~\eqref{eq:slb} gives the first branch. For the restricted output alphabet $\{00,11\}$, transition inputs contribute $2C$ to the expected number of errors per pair. Conditional on a constant input, $11$ has probability $B/S$, and each wrong reproduction costs two errors. The conditional bit error allowance is therefore $(2D-2C)/(2S)=\delta_B$. Let $C_{\rm same}$ indicate that the input is constant. The chain rule $I(X;Y)=I(C_{\rm same};Y)+I(X;Y\mid C_{\rm same})$ bounds the rate from below by $S[\hx(B/S)-\hx(\delta_B)]$. The optimal Bernoulli channel on the constant inputs attains this bound if each input containing a flip independently draws from its reproduction marginal. Indeed, $Y$ is then independent of $C_{\rm same}$.

The Bernoulli backward channel at crossover $\delta_B$ has reproduction probabilities
\[
 q_{00}=\frac{A-Bu^2}{S(1-u^2)},\qquad q_{11}=\frac{B-Au^2}{S(1-u^2)} .
\]
Equation~\eqref{eq:Z} gives \eqref{eq:v2b}. Both transition columns satisfy \eqref{eq:kkt} if $2Sx+C/x\le1$, where $x=u/(1+u^2)$. The two roots of $2Sx^2-x+C=0$ are $C/S$ and $1/2$. At the first loss $x=C/S$, and $x$ increases towards $1/2$ as $u$ increases. The two active masses stay positive until $\delta_B=\min(A,B)/S$, which is $D=D_0$. This proves optimality over all four possible reproduction words.
\end{proof}

For $(a,b)=(0.1,0.3)$, the full-support threshold is $0.08167$ and $D_0=0.25$. At $D=0.125$, transition inputs consume $0.15$ of the $0.25$ errors allowed per pair. The constant inputs have total probability $0.85$, giving $\delta_B=1/17$. Their two reproduction probabilities are $5/6$ and $1/6$. Equations~\eqref{eq:r2} and \eqref{eq:v2b} give $r_2=0.17458$ bits and $v_2=0.38041$ bits$^2$ per symbol.

\subsection{Symmetric sources}\label{sec:sym}
Suppose $a=b=p$ with $0<p<1/2$, so $\pi_0=\pi_1=1/2$ and $\rho=1-2p$. A word with $t$ flips has probability $p^t(1-p)^{K-1-t}/2$. The number of flips $T$ in $X^K$ is binomial with parameters $K-1,p$. Put $\nu(p)=p(1-p)\log_2^2((1-p)/p)$. Wherever \eqref{eq:slb} holds, we obtain
\begin{align}
 R_K(KD)&=1+(K-1)\hx(p)-K\hx(D),\qquad V_K(KD)=(K-1)\nu(p),\label{eq:lowp}\\
 \jj_K(X,D)-R_K(KD)&=(T-(K-1)p)\log_2((1-p)/p).               \label{eq:bin}
\end{align}
After a support change, words with the same number of flips can have different tilted information. Their normalizers depend on their distances to the remaining reproduction words and on the probabilities assigned to those words.

Complementing both words or reversing both words preserves their Hamming distance and the source law. The uniqueness of $q$ implies that words related by these operations have equal reproduction probabilities. Group the source words into classes with total probabilities $s_i$ and common normalizers $Z_i$. For any such grouping,
\begin{align}
 KD&=u\sum_i s_i\,\partial_u\ln Z_i,\qquad
 R_K(KD)=KD\log_2u-\sum_i s_i\log_2Z_i,                        \label{eq:class}\\
 V_K(KD)&=\Vfun(s,Z):=\sum_{i<j}s_i s_j\log_2^2(Z_i/Z_j).       \label{eq:pairvar}
\end{align}
The derivative holds the reproduction masses fixed. The last identity equals half the expected squared difference between two independent draws of $-\log_2 Z(X)$. Every input class contributes, including a class whose reproduction mass has vanished.

For comparison, at $K=2$ the two constants are equiprobable after the first support threshold. With $\delta_B=(2D-p)/(2(1-p))$, Theorem~\ref{thm:two} gives
\begin{equation}
 R_2(2D)=(1-p)[1-\hx(\delta_B)],\qquad
 V_2(2D)=p(1-p)\log_2^2\!\frac{1}{2\sqrt{\delta_B(1-\delta_B)}} .             \label{eq:sym2}
\end{equation}
All symmetric-source results extend to $1/2<p<1$ by XORing the input and reproduction with $0101\ldots$. This distance-preserving bijection changes $p$ to $1-p$. At $p=1/2$ the rate is $K[1-\hx(D)]$ and the variance is zero for $0\le D\le1/2$. At $p=0$ or $1$, one fair bit determines the whole block, the rate is $1-\hx(D)$, and the variance is zero on the same interval. Encoding that bit through a Bernoulli test channel attains the rate. Averaging the single-coordinate conditional entropy bounds proves its optimality.

\subsubsection{Three symbols}\label{sec:three}
There are three classes, indexed by the number of flips. The constants $000,111$ have total probability $s_0=(1-p)^2$. The four words with one flip have probability $s_1=2p(1-p)$ in total, and the alternating words $010,101$ have total probability $s_2=p^2$. Write $S_{01}=s_0+s_1$.
\begin{theorem}[Three-symbol rate and variance]\label{thm:three}
The support has eight words for $0\le D<D_{\rm FS}(3)$, six for $D_{\rm FS}(3)\le D<D_2$, and two for $D_2\le D<1/2$. The alternating words leave at $D_{\rm FS}(3)$ and the one-flip words at $D_2$, where
\begin{align}
 D_{\rm FS}(3)&=\frac{1-\sqrt{1-4p^2}}2,\qquad u_1=\frac{D_{\rm FS}(3)}{1-D_{\rm FS}(3)},\qquad
 u_2=\frac{2p}{1+\sqrt{1-4p^2}},                               \label{eq:d3a}\\
 D_2&=\frac13\left[s_0\frac{3u_2^3}{1+u_2^3}
              +(1-s_0)\left(1+\frac{u_2}{1+u_2}\right)\right].  \label{eq:d3b}
\end{align}
\normalfont
For $0\le D\le D_{\rm FS}(3)$ the two quantities are $R_3=1+2\hx(p)-3\hx(D)$ and $V_3=2\nu(p)$. On the other two branches, with $s=(s_0,s_1,s_2)$, they are
\[
 R_3(3D)=3D\log_2u-\sum_{i=0}^2s_i\log_2Z_i,\qquad V_3(3D)=\Vfun(s,Z).
\]
For $D_{\rm FS}(3)<D<D_2$, solve for the unique $u\in(u_1,u_2)$ satisfying
\begin{equation}
 3D=s_2+\frac{u}{1+u}+S_{01}\frac{u(1+2u)}{1+u+u^2}.          \label{eq:3mid}
\end{equation}
Put $\beta_u=(u+u^2)/2$ and $H_u=(1+u)(1+u+u^2)/2$. The three normalizers in class order are $(s_0 H_u/S_{01},\,s_1 H_u/(2S_{01}),\,\beta_u)$.

For $D_2\le D<1/2$, the two constants have equal reproduction mass. A constant input is at distances $0,3$ from the two outputs, and every other input is at distances $1,2$. Thus $\alpha_u=(1+u^3)/2$ is the normalizer for constant inputs and $\beta_u$ for the others. The unique $u\in[u_2,1)$ satisfies
\[
 3D=s_0\frac{3u^3}{1+u^3}+(1-s_0)\left(1+\frac{u}{1+u}\right).
\]
In particular $V_3(3D)=s_0(1-s_0)\log_2^2(\alpha_u/\beta_u)$. Clearing denominators gives a cubic on each branch. The distortion increases strictly with $u$ on that branch. \end{theorem}
The derivation is given in Appendix~\ref{app:allocation}. At $p=0.1$, the two thresholds are $0.010102$ and $0.069978$.

\subsubsection{Four symbols}\label{sec:four}
At length four the number of flips alone no longer describes the classes. For example, both $0001$ and $0011$ have one flip. With only the constant words as reproductions, their distance pairs are $(1,3)$ and $(2,2)$, so their tilted information differs. Complement and reversal give six classes. Table~\ref{tab:four} groups them by parity (the parity of the number of ones) and Figure~\ref{fig:ladder} gives their support thresholds for $p=0.1$.

\begin{table}[htbp]\centering\small
\caption{Four-symbol classes for $0<p<1/2$. Each representative denotes its orbit under complement and reversal.}\label{tab:four}
\begin{tabular}{@{}llcl@{}}\toprule
Parity & Representative & Class size & Source probability $s_i$\\\midrule
Even & $0000$ & 2 & $(1-p)^3$\\
     & $0011$ & 2 & $p(1-p)^2$\\
     & $0110$ & 2 & $p^2(1-p)$\\
     & $0101$ & 2 & $p^3$\\
Odd  & $0001$ & 4 & $2p(1-p)^2$\\
     & $0010$ & 4 & $2p^2(1-p)$\\\bottomrule
\end{tabular}
\end{table}

\begin{figure}[htbp]\centering
\includegraphics[width=1\linewidth]{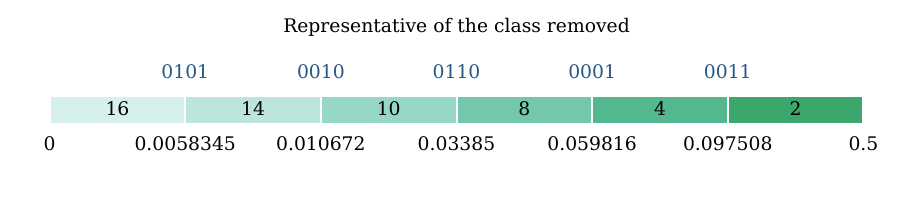}
\caption{Four-symbol support at $p=0.1$. The number inside each interval is the number of active words. At each labeled threshold the indicated class leaves the support. Interval widths are schematic; the lower labels give the actual distortions.}\label{fig:ladder}
\end{figure}

At $p=0.1$, $D=0.10$ the optimal reproduction law uses only the two constant words at all three lengths $K=2,3,4$. The corresponding rates are $0.31071$, $0.23530$, and $0.20335$ bits per symbol. The variances increase with length at this distortion, from $0.05707$ to $0.09770$ to $0.24684$ bits$^2$ per symbol. Figure~\ref{fig:small} gives the full curves.

At a trial multiplier $u$, divide the reproduction mass between the parity groups, assigning $\theta$ to the even group and $1-\theta$ to the odd group. For each proposed $\theta$, the allocation among classes within each group is explicit. The inner search chooses the best $\theta$ for that $u$. The outer search varies $u$, repeating the allocation until the resulting channel meets the distortion budget.

To compute the allocation, average $u^{d(x,y)}$ over a complementary output pair. The average is $\alpha_u=(1+u^4)/2$ if the pair contains $x$, $\beta_u=(u+u^3)/2$ if its parity differs from that of $x$, and $\chi_u=u^2$ otherwise. Each even class contains one pair and each odd class two. If $w_i$ is the reproduction mass of class $i$ in parity group $g$, its normalizer is
\begin{equation}
 Z_i=L_g+k_gw_i,\qquad m_E=\theta,\quad m_O=1-\theta,\qquad
 L_g=\chi_u m_g+\beta_u(1-m_g),\label{eq:fourz}
\end{equation}
where $\theta$ is the total reproduction mass of the even group, $k_E=\alpha_u-\chi_u$ and $k_O=k_E/2$.

\begin{theorem}[Four-symbol rate and variance]\label{thm:four}
For $0<D<1/2$, the rate and variance are
\[
 R_4(4D)=4D\log_2u-\sum_i s_i\log_2Z_i,\qquad V_4(4D)=\Vfun(s,Z).
\]
Here $Z_i$ is given by \eqref{eq:fourz}, with the class probabilities in Table~\ref{tab:four} listed in decreasing order within each parity group. At multiplier $u$ and total even mass $\theta$, let $\eta_g$ be the multiplier for the mass constraint in group $g$. This multiplier and the optimal class masses are
\begin{equation}
 \eta_g=\max_{1\le j\le |g|}\frac{\sum_{i=1}^j s_{g,i}}{m_g+jL_g/k_g},\qquad
 w_{g,i}=[s_{g,i}/\eta_g-L_g/k_g]_+,                          \label{eq:water}
\end{equation}
where $[z]_+=\max(0,z)$. Within each group, the active classes are the first $j$ in decreasing order of source probability, with $j$ equal to the number of positive masses. At $m_g=0$ the maximum in \eqref{eq:water} is attained at $j=1$, so $\eta_g=s_{g,1}k_g/L_g$ and all $w_{g,i}=0$. This supplies the one-sided derivative for an empty group at either endpoint. After optimizing those masses, the derivative of the objective $\sum_i s_i\ln Z_i$ with respect to $\theta$ is
\begin{equation}
 G(\theta)=\eta_E-\eta_O+(\chi_u-\beta_u)
      \left(\sum_{i\in E}\frac{s_i}{Z_i}-\sum_{i\in O}\frac{s_i}{Z_i}\right).
                                                                     \label{eq:G}
\end{equation}
Concavity gives the maximizer at $\theta=0$ if $G(0)\le0$, at $\theta=1$ if $G(1)\ge0$, and at a root $G(\theta)=0$ otherwise. With this allocation at each trial $u$, choose $u$ to satisfy the distortion equation
\begin{equation}
 4D=u\sum_{g}\sum_{i\in g}s_i\frac{L_g'+k_g'w_i}{Z_i},          \label{eq:fourD}
\end{equation}
where primes differentiate in $u$ with $\theta,w_i$ held fixed. These conditions determine the optimum at each $u$. At a threshold a class mass is zero and its column value equals one. The allocations on either side determine whether the class enters or leaves.
\end{theorem}
Appendix~\ref{app:allocation} derives the allocation and the endpoint rules.
For $p=0.1$, the support size decreases as $16,14,10,8,4,2$. For general $0<p<1/2$, the two-constant candidate must pass two tests at the trial $u$. With $\theta=1$, all even mass goes to the first class precisely when the second even class has no greater marginal gain, $s_{E,2}/\chi_u\le s_{E,1}/\alpha_u$. By Table~\ref{tab:four}, this is $p/(1-p)\le 2u^2/(1+u^4)$. The other test, $G(1)\ge0$, ensures optimality against a transfer to the odd group. Together the tests characterize the two-constant solution without assuming a common sequence of support losses for all $p$.

On this branch, the total input probabilities of the constant, odd-weight, and weight-two words are $S_0=(1-p)^3$, $S_1=2p(1-p)$, and $S_2=1-S_0-S_1$. Their normalizers are $(\alpha_u,\beta_u,\chi_u)$, and
\[
 4D=S_0\frac{4u^4}{1+u^4}+S_1\left(1+\frac{2u^2}{1+u^2}\right)+2S_2 .
\]

\begin{figure}[!htb]\centering
\includegraphics[width=1\linewidth]{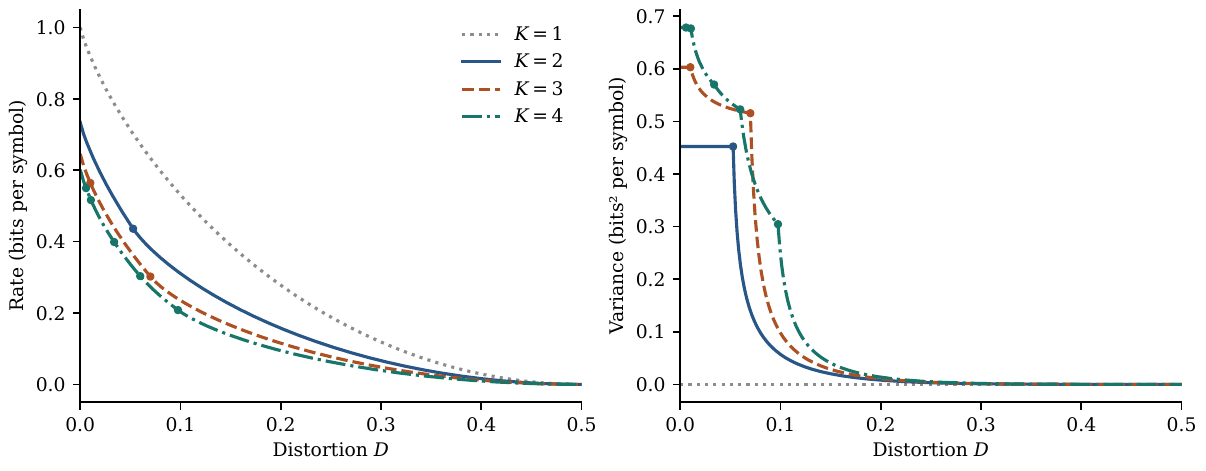}
\caption{Symmetric chain with $p=0.1$. Markers on the $K=2,3,4$ curves locate the support thresholds. The first marker ends each variance plateau. Later thresholds alter the optimal masses, even where the change in the plotted curve is slight.}\label{fig:small}
\end{figure}

\section{Computing rate and variance at arbitrary block length}\label{sec:general}
For $K=2,3,4$ the support of the optimal reproduction law led to explicit formulas for both moments. At larger $K$ there are two ways to obtain them. One can solve the $K$-symbol problem itself, exploiting the product structure of Hamming distortion. When this becomes too expensive, shorter blocks give a guaranteed interval for the rate and, for symmetric persistent sources, for the variance.

\subsection{One optimization gives both moments}\label{sec:compute}
Fix a positive binary transition matrix and $0<D<D_0$. The source assigns probabilities $P_K(x)$ to the $2^K$ words. At a trial multiplier $u\in(0,1)$, let $A_u(x,y)=u^{d(x,y)}$. Blahut's recursion for the binary Hamming matrix \cite[proof of Cor.~7]{blahut} gives the tensor factorization
\[
 A_u=\begin{pmatrix}1&u\\ u&1\end{pmatrix}^{\otimes K}.
\]
The optimum maximizes $\sum_x P_K(x)\ln (A_uq)(x)$ over reproduction probabilities $q$. From any positive initial $q$, the Blahut--Arimoto iteration takes the form \cite{blahut}
\[
 Z=A_u q,\qquad q_{\rm new}=q\odot A_u(P_K/Z),
\]
with elementwise multiplication and division. Each tensor multiplication visits the $K$ bit positions in turn. At a given position it replaces each pair of entries $(z_0,z_1)$ differing there by $(z_0+uz_1,uz_0+z_1)$. A dense multiplication sums over $2^K$ reproduction words for each of $2^K$ source words, requiring $O(4^K)$ arithmetic operations. Applying the binary factors reduces the cost per Blahut--Arimoto iteration to $O(K2^K)$ operations and uses $O(2^K)$ storage. The Markov dependence enters through the weights $P_K$.

The distortion constraint therefore selects the value of $u$ for which $\sum_x P_K(x)\,u\,\partial_u\ln Z(x)=KD$, holding the optimizing $q$ fixed in the derivative. At that solution the normalizers already determine the tilted information in \eqref{eq:j}, so its variance requires only a weighted sum over source words. Let $\ell(x)=-\log_2 Z(x)$ and $\bar\ell=\sum_x P_K(x)\ell(x)$. Then
\begin{equation}
 \boxed{\quad r_K(D)=D\log_2u+\frac{\bar\ell}{K},\qquad
 v_K(D)=\frac1K\sum_x P_K(x)(\ell(x)-\bar\ell)^2 .\quad}
 \label{eq:directmoments}
\end{equation}
Multiplying either expression by $K$ gives its block value. Equation~\eqref{eq:directmoments} is exact at the optimizer. A finite iteration gives a numerical estimate. Its accuracy can be assessed from the attained distortion and all $2^K$ column inequalities in \eqref{eq:kkt}, including unused words. Appendix~\ref{app:certificate} converts an upper bound on the largest column residual into error bounds at a given multiplier. At $D=0$ the two block quantities are entropy and varentropy. At and above $D_0$ we use the zero-slope convention of Section~\ref{sec:setup}.

\paragraph{A faster evaluation when only constants are reproduced.}
For a symmetric chain on a verified branch with reproduction support $\{0^K,1^K\}$, the input Hamming weight $w$ alone determines $Z_w=(u^w+u^{K-w})/2$. Let $L_i(w,b)$ be the joint probability of weight $w$ and last bit $b$ after $i$ symbols. Conditioning on the final transition gives the recursion, with initial masses $L_1(0,0)=L_1(1,1)=1/2$,
\[
 L_{i+1}(w,b)=\sum_{a=0}^1 L_i(w-b,a)P_{ab},\qquad
 s_w=L_K(w,0)+L_K(w,1).
\]
All unspecified initial masses and out-of-range entries are zero. The $K+1$ weights $s_w$ and normalizers $Z_w$ then suffice in \eqref{eq:directmoments}, summing over $w$ instead of words. The weight law takes $O(K^2)$ time and $O(K)$ storage. Each subsequent multiplier evaluation takes $O(K)$ time. To establish optimality of this two-word candidate, all column inequalities in \eqref{eq:kkt} must hold.

\subsection{The full-support interval for a symmetric chain}
Fix $0<p<1/2$ and $0<D<1/2$, and write $z=1-2D$ and $\rho=1-2p$. The inverse noise channel in \eqref{eq:inverse} can assign negative weights. We test feasibility by tracking the mass it assigns to each reproduction prefix. In spin notation $x,y\in\{-1,1\}$, a single inverse-channel weight is $(1+xy/z)/2$. For a fixed prefix $y_1,\ldots,y_{i-1}$, define
\[
 W_{i-1}=\prod_{j=1}^{i-1}\frac{1+X_j y_j/z}{2},\qquad
 m=\E W_{i-1},\qquad \mu=\frac{\E[X_i W_{i-1}]}{m}\quad(m>0),
\]
where the expectation uses the original Markov law in spin coordinates. Here $m$ is the candidate prefix mass and $\mu$ is the normalized signed first moment of the next spin. Appending $y$ gives mass $m(1+y\mu/z)/2$. Both extensions therefore have nonnegative mass precisely when $|\mu|\le z$. The largest $|\mu|$ at each length suffices for the feasibility test.
\begin{proposition}[Full-support threshold]\label{prop:rec}
Let $H_i$ be the largest $|\mu|$ over positive-mass reproduction prefixes of length $i-1$. The empty prefix gives $H_1=0$. While $H_i<z$,
\begin{equation}
 H_{i+1}=\rho\frac{1-zH_i}{z-H_i}.                             \label{eq:rec}
\end{equation}
The backward Bernoulli$(D)$ channel is feasible exactly when $H_i<z$ for $i<K$ and $H_K\le z$. The finite-block thresholds decrease to Gray's critical distortion \cite{gray70,jalali}
\begin{equation}
 D_c=\frac12\left(1-\frac{2\sqrt\rho}{1+\rho}\right)
     =\frac12\left(1-\sqrt{1-\left(\frac p{1-p}\right)^2}\right).
                                                                    \label{eq:Dc}
\end{equation}
\end{proposition}
For a child of positive mass, multiplying by the next inverse-channel weight and predicting through one Markov transition gives $\mu'=\rho(\mu+y/z)/(1+y\mu/z)$. For $|\mu|<z$ the opposite spin gives the largest $|\mu'|$, namely the right side of \eqref{eq:rec} with $H_i=|\mu|$. This function increases with $|\mu|$, and the alternating prefix attains the bound at every step. If equality occurs before the last step, the zero-mass child still has nonzero signed first moment, so one of its children has negative mass.

The fixed point equation is $H^2-z(1+\rho)H+\rho=0$. If $z\ge2\sqrt\rho/(1+\rho)$, iteration from zero increases to the smaller root below $z$. Below that threshold the iteration must reach or cross $z$ in finitely many steps. An increasing orbit below $z$ could only approach a fixed point or the pole, where the next iterate diverges. Feasibility for length $K+1$ implies feasibility for its $K$-symbol marginal, and feasible distortions form an interval by composition of binary symmetric channels. This proves the proposition. Figure~\ref{fig:recursion} illustrates the two possibilities.

\begin{figure}[htbp]\centering
\includegraphics[width=.82\linewidth]{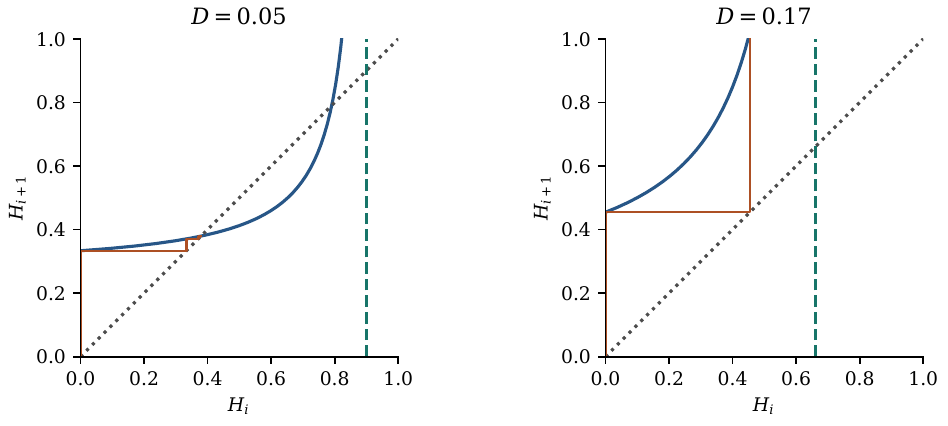}
\caption{Scalar recursion for $p=0.35$ ($D_c\simeq0.078675$). From $H_1=0$ the orange orbit follows the update curve (blue) and the diagonal (dotted). Below $D_c$ the orbit approaches a fixed point. Above $D_c$ it crosses the feasibility boundary $H=z$ (dashed).}\label{fig:recursion}
\end{figure}

For $p=0.1$ the thresholds at lengths $2,3,4$ are $0.05279$, $0.01010$, and $0.00583$, with limit $D_c=0.00310$. On the common interval $0\le D\le D_c$, \eqref{eq:lowp} gives $v_K(D)=(1-1/K)\nu(p)$ at every $K$. Testing a single $D$ requires only $O(K)$ steps. For arbitrary bias, nonnegativity of the inverse in \eqref{eq:inverse} still characterizes this branch.

\subsection{Varentropy along full-support branches}
The next result applies to positive binary chains of arbitrary stationary bias. At each length choose any distortion on the branch where \eqref{eq:slb} holds.
\begin{theorem}[Finite-block varentropy]\label{thm:varentropy}
Write $B_0=-\log_2\pi_{X_1}$ and $A_t=-\log_2P_{X_tX_{t+1}}$. Then block self-information is $B_0+\sum_{t=1}^{K-1}A_t$. Define $N=K-1$, $s^2=\Var A_1$, $\gamma=\Cov(A_1,A_2)$, and $\kappa=\Cov(B_0,A_1)$. The block variance $V_K=V_K(KD)$ satisfies
\begin{equation}
 V_K=\Var B_0+N s^2+
       2\gamma\sum_{h=1}^{N-1}(N-h)\rho^{h-1}
       +2\kappa\sum_{h=0}^{N-1}\rho^h .                              \label{eq:varent}
\end{equation}
Empty sums vanish. In particular, along any sequence of distortions on the full-support branches, $V_K/K\to\sigma_X^2=s^2+2\gamma/(1-\rho)$. For $0<|\rho|<1$ the first correction follows by summing geometric series, giving
\[
 \frac{V_K}{K}=\sigma_X^2+\frac{c_*}{K}+O\!\left(\frac{|\rho|^{K-1}}K\right),\qquad
 c_*=\Var B_0-s^2+\frac{2\kappa}{1-\rho}
 -2\gamma\left(\frac1{1-\rho}+\frac1{(1-\rho)^2}\right).
\]
\end{theorem}
\begin{proof}
For a two-state stationary chain, conditioning on the last state of the first edge and applying $P^h$ gives
$\Cov(A_1,A_{1+h})=\rho^{h-1}\gamma$ for $h\ge1$ and $\Cov(B_0,A_{1+h})=\rho^h \kappa$ for $h\ge0$. Expanding the variance of the sum gives \eqref{eq:varent}; summing the two finite geometric series gives the displayed correction. For $\rho=0$ the bits are independent and the normalized varentropy is constant.
\end{proof}
Figure~\ref{fig:varent} shows the correction for a biased persistent chain.

\begin{figure}[htbp]\centering
\includegraphics[width=.94\linewidth]{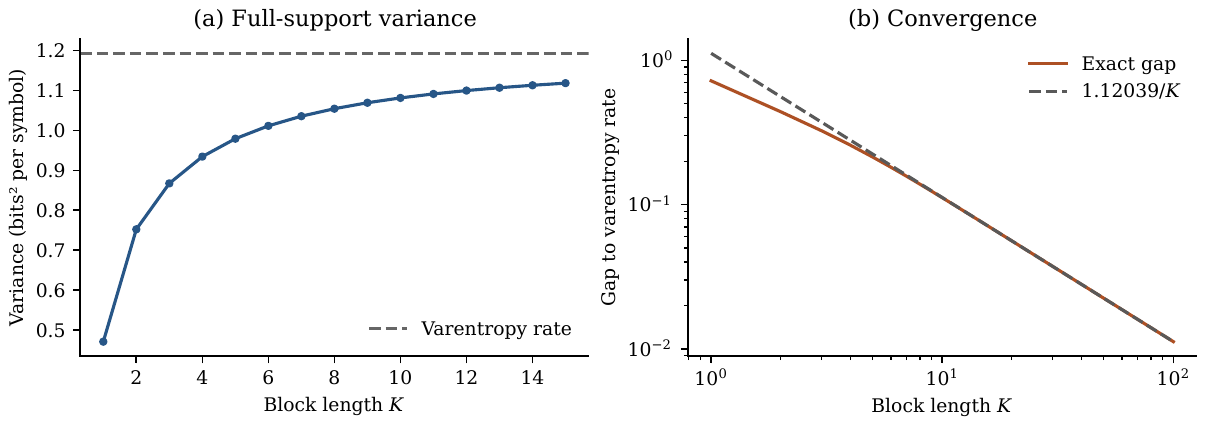}
\caption{Full-support variance for $(a,b)=(0.1,0.3)$, evaluated from \eqref{eq:varent} for any $D_K$ on the full-support branch at length $K$. The variance per symbol approaches $1.1930$ bits$^2$ per symbol (left), with a gap asymptotic to $1.1204/K$ (right). The comparison is from \cite{itw}.}
\label{fig:varent}
\end{figure}

For this chain the full-support threshold falls from $0.08167$ at $K=2$ to about $0.01235$ at $K=15$. The values in Figure~\ref{fig:varent} thus cannot be used to infer the fixed-distortion limit at, for example, $D=0.125$. At a fixed distortion each block length must be evaluated at its own optimal reproduction law and slope.

\subsection{Approximating the rate from shorter blocks}
Put $I_0=I(X_1;X_2)=H(X_1)-H(X_2\mid X_1)$. A cut between successive blocks removes precisely $I_0$ bits of source dependence. For $K=mk$ this gives the finite-block bounds \cite{wynerziv,blahut,jalali}
\begin{equation}
 r_k(D)-\frac{(1-1/m)I_0}{k}\ \le\ r_K(D)\ \le\ r_k(D).        \label{eq:bracket}
\end{equation}
For the upper bound, apply an optimal $k$-symbol channel to each segment, independently conditional on its input. Conditional output entropies add, and total output entropy is at most their sum. For the lower bound, the conditional entropy of the entire input given its reproduction is at most the sum of the corresponding segment conditional entropies. The segment entropies sum to $H(X^K)+(m-1)I_0$. Convexity of $R_k$ allows an arbitrary allocation of the distortion among segments.

Letting $m$ grow in \eqref{eq:bracket} gives $0\le r_k(D)-R(D)\le I_0/k$. For unequal segment lengths $k_j$, let $\mathcal R=\min\sum_{j=1}^mR_{k_j}(d_j)$, where $d_j\ge0$ and $\sum_jd_j\le KD$. The same proof gives $\mathcal R-(m-1)I_0\le R_K(KD)\le \mathcal R$. With length-$k$ segments and a possible shorter remainder the per-symbol width is at most $I_0/k$. Thus a choice $k\ge I_0/\delta$ gives a bracket of width at most $\delta$ bits per symbol whenever $K\ge k$; for $K<k$ one solves the original problem. For example, with $p=0.1$, $I_0=1-\hx(0.1)\simeq0.531$. If $K=80$ and $k=8$, a single eight-symbol rate calculation brackets $r_{80}(D)$ to within a width of $0.0598$ bits per symbol at every distortion. The rate on the full-support branch is more precise still. The Markov chain rule and \eqref{eq:slb} give $r_K(D)=H(X_2\mid X_1)-\hx(D)+I_0/K$ exactly.

\section{The dispersion conjecture and its evidence}\label{sec:conjecture}
Now apply the excess-distortion criterion defining $R^*(n,D,\epsilon)$ to one stationary Markov path of length $n$. The marginal of any consecutive $K$ symbols is still $P_K$. The first-order cost of discarding the dependence across a block boundary is at most $I(X_1;X_2)$ bits (Section~\ref{sec:general}). The dispersion question asks whether taking $m\to\infty$ for each fixed $K$ in \eqref{eq:oper}, then letting $K$ increase, gives the same second-order coefficient as coding the Markov source itself.
For a stationary symmetric chain with flip probability $0<p<1/2$, Theorem~\ref{thm:common} proves that the normalized block variance has a positive limit at each interior differentiability point $D$ of $R$. Denote this limit by
\begin{equation}
 \Vhier(D)=\lim_{K\to\infty} V_K(KD)/K.                          \label{eq:hier}
\end{equation}
The remaining question concerns the rate for a single path.

\begin{conjecture}[Operational dispersion]\label{conj:hier}
For a stationary symmetric binary Markov source with $0<p<1/2$ under Hamming distortion, fix $0<D<1/2$ where $R$ is differentiable and $0<\epsilon<1$. As $n\to\infty$, the minimum rate for a single path satisfies
\begin{equation}
 R^*(n,D,\epsilon)=R(D)+\sqrt{\frac{\Vhier(D)}{n}}\,Q^{-1}(\epsilon)
                          +o(n^{-1/2}).                          \label{eq:conj}
\end{equation}
\end{conjecture}

For an i.i.d.\ Bernoulli source, the block tilted information is additive, so every $v_K$ equals the known single-letter dispersion. For symmetric chains on the interval $D\le D_c$ the exact identity $v_K=(1-1/K)\nu(p)$ gives a second analytical check. Outside full support the finite-length computations provide empirical comparisons.

For a symmetric persistent chain at $0<D\le D_c$, the tilted information of one path is an affine function of the number of flips. Its Gaussian limit gives the converse coefficient in \eqref{eq:conj}. An achievability theorem at that coefficient remains open.
Figure~\ref{fig:logic} separates the arguments for the variance limit and the operational converse.
\begin{figure}[htbp]\centering
\includegraphics[width=.99\linewidth]{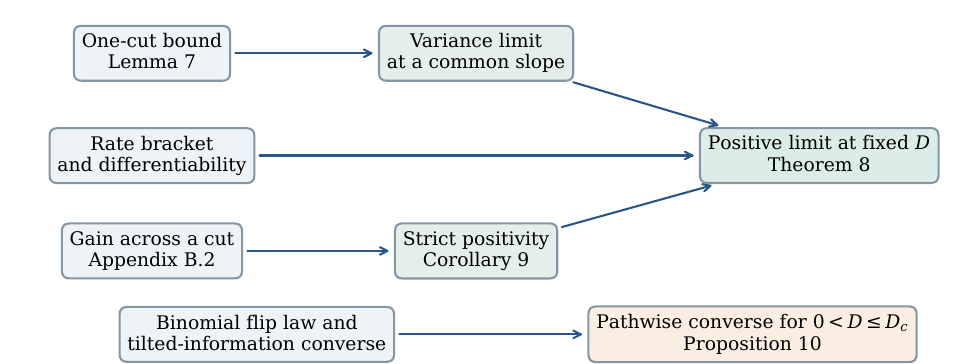}
\caption{Dependencies of the limiting results. Only the bottom row uses a coding inequality for the Markov source itself.}\label{fig:logic}
\end{figure}
\subsection{Variance convergence at fixed distortion}\label{sec:variance}
Consider a symmetric chain with $0<p<1/2$. At $D\le D_c$, equation~\eqref{eq:lowp} already gives the limit $\nu(p)$. Beyond that interval, changes of reproduction support preclude the same calculation. We first compare blocks at a common distortion multiplier and then show how the optimizing multipliers behave at a prescribed distortion.

The cut comparison uses natural logarithms. Write $\tau=-\ln u$ in nats per error. For $k$ symbols let $Z_{q,\tau}(x)=\sum_y q(y)e^{-\tau d(x,y)}$. The only term of the tilted information that depends on $x$ is $-\ln Z_{q,\tau}(x)$. We call this term the \emph{score at multiplier $\tau$}. For an optimal marginal $q^*$ define
\begin{equation}
 \begin{aligned}
 F_k(\tau)&=\min_q\E[-\ln Z_{q,\tau}(X^k)],& f_k(\tau)&=F_k(\tau)/k,\\
 \ell_{k,\tau}(x)&=-\ln Z_{q^*,\tau}(x),& \omega_k(\tau)&=\sd(\ell_{k,\tau}(X^k))/\sqrt{k}.
 \end{aligned}                                                    \label{eq:canon}
\end{equation}
$F_k(\tau)$ is the minimum of mutual information in nats plus $\tau\E d$. At a multiplier $\tau_K$ attaining distortion $D$, the constant term in \eqref{eq:j} disappears on centering, so $v_K(D)=\omega_K(\tau_K)^2/(\ln2)^2$.

\begin{lemma}[One cut at a common multiplier]\label{lem:cut}
At $\tau\ge0$, split a symmetric stationary path into adjacent words $X_{\rm L}=X_1^i$ and $X_{\rm R}=X_{i+1}^{i+j}$. Under its true joint law, $\ell_{i,\tau}(X_{\rm L})$ and $\ell_{j,\tau}(X_{\rm R})$ are independent, and
\[
 \E\!\left[\big(\ell_{i+j,\tau}(X_{\rm L}X_{\rm R})-\ell_{i,\tau}(X_{\rm L})-\ell_{j,\tau}(X_{\rm R})\big)^2\right]\le\rho .
\]
\end{lemma}
Complement symmetry makes the tilted information values of the two shorter blocks functions of disjoint sets of independent flips. Appendix~\ref{app:split} bounds the cost of the omitted boundary flip.
\begin{theorem}[A variance limit for symmetric sources]\label{thm:common}
There is a continuous function $\omega_\infty(\tau)$, zero at $\tau=0$ and positive at every $\tau>0$, such that
\begin{equation}
 |\omega_k(\tau)-\omega_\infty(\tau)|\le(1+\sqrt2)\sqrt{\rho/k},\qquad \tau\ge0 .  \label{eq:common}
\end{equation}
For $0<D<1/2$ at which $R(D)$ is differentiable, let $\tau(D)=-(\ln2)R'(D)$. Then
\begin{equation}
 \lim_{K\to\infty}v_K(D)=\frac{\omega_\infty(\tau(D))^2}{(\ln2)^2}>0.  \label{eq:fixed}
\end{equation}
\end{theorem}
\begin{proof}
Applying the one-cut bound \eqref{eq:cut} to two blocks of length $k$ shows that doubling the block changes its normalized standard deviation by at most $\sqrt{\rho/(2k)}$. These errors are summable over successive doublings. Comparing common multiples gives \eqref{eq:common}. Appendix~\ref{app:positive} proves positivity by relating the gain from joining blocks to the covariance of the score with the number of flips.

We now fix the distortion $D$. Its optimizing multiplier can change with block length. Put $J_0=(1-\hx(p))\ln 2$. Applying the boundary comparison in \eqref{eq:bracket} to the penalized objective gives
\begin{equation}
 f(\tau)=\lim_{k\to\infty} f_k(\tau),\qquad 0\le f_k(\tau)-f(\tau)\le J_0/k .
 \label{eq:redun}
\end{equation}
The finite-block multipliers maximize $f_K(\tau)-\tau D$, and $\max_\tau[f(\tau)-\tau D]=R(D)\ln2$. The uniform reproduction prior gives $Z(x)=[(1+e^{-\tau})/2]^K$, hence $f_K(\tau)\le\ln2$. Since the objective at an optimum is at least its value zero at $\tau=0$, every maximizing slope lies in $[0,(\ln2)/D]$. Thus every limit point maximizes $f(\tau)-\tau D$. Differentiability gives a unique maximizer $\tau(D)$, which is positive in the interior by the rate bound in Appendix~\ref{app:positive}. Taking $\tau_K\to \tau(D)$ in the uniform estimate \eqref{eq:common} proves \eqref{eq:fixed}.
\end{proof}

\paragraph{At a possible corner of the limiting rate.}
A change of finite-block reproduction support need not produce a corner in $R$. At an interior corner $D$, let $\partial R(D)=[R'_-(D),R'_+(D)]$ be the interval of supporting slopes in bits per error. The compactness argument above still applies. Each limit point $\tau_*$ of the finite-block multipliers maximizes $f(\tau)-\tau D$, so $-\tau_*/\ln2\in\partial R(D)$. Along any subsequence with $\tau_K\to\tau_*$, continuity and \eqref{eq:common} give $v_K(D)\to\omega_\infty(\tau_*)^2/(\ln2)^2$. A unique variance limit at that corner would follow if all attained multiplier limits gave the same value. Convexity allows at most countably many interior corners.

\begin{corollary}[Quantitative positivity]\label{cor:positive}
For $\tau>0$ the limiting normalized standard deviation obeys
\begin{equation}
 \omega_\infty(\tau)\ge\frac{\rho\sqrt{p(1-p)}}2
            \left[\frac{2\ln\cosh(\tau/2)}\tau\right]^4 .                \label{eq:lower}
\end{equation}
\end{corollary}
Appendix~\ref{app:positive} proves this bound by comparing the information gained across a cut with the covariance of the block score and its number of flips.

\paragraph{Symmetric anti-persistent sources.}
For $1/2<p<1$, XOR both source and reproduction paths with $0101\ldots$, as in Section~\ref{sec:sym}. The mask takes each flip to a repetition and preserves Hamming distance. It therefore permutes the optimal reproduction laws and their normalizers at each multiplier. A sub-block starting on the opposite mask phase differs by a global complement, which leaves its score unchanged. Applying Lemma~\ref{lem:cut} and Theorem~\ref{thm:common} to the transformed chain with flip probability $1-p$ gives the same conclusions, including the variance enclosure below, with $|\rho|$ in place of $\rho$. For a biased chain the block score can also depend on its initial state. The two scores across a cut can then be correlated, and extending the proof requires control of that covariance.
\subsection{From shorter blocks to a variance enclosure}\label{sec:shortvariance}
Theorem~\ref{thm:common} compares standard deviations at a common multiplier. At fixed distortion the optimal multiplier can change with block length. We first use a short block to locate a set containing every possible long-block optimum. We then compute the smallest and largest short-block normalized standard deviations on that set. Finally, we widen this range by the error from the cut comparison before squaring to obtain variance bounds.

For a probe of length $b\le K$, write $g_b(\tau)=f_b(\tau)-\tau D$ on $[0,\ln2/D]$. Let $I_b$ consist of multipliers for which $g_b$ lies within $2J_0/b$ of its maximum, with $J_0=(1-\hx(p))\ln2$. By \eqref{eq:redun} this set contains every optimal length-$K$ multiplier. Applying \eqref{eq:common} at both lengths with normalized standard-deviation allowance $\Delta_{K,b}=(1+\sqrt2)\sqrt\rho\,(b^{-1/2}+K^{-1/2})$ gives
\begin{equation}\label{eq:cert}
 \frac{[\min_{\tau\in I_b}\omega_b(\tau)-\Delta_{K,b}]_+^2}{(\ln2)^2}
 \ \le\ v_K(D)\ \le\ 
 \frac{(\max_{\tau\in I_b}\omega_b(\tau)+\Delta_{K,b})^2}{(\ln2)^2}.
\end{equation}
If the interval width is at most $2\delta_v$, its midpoint gives $v_K(D)$ to additive accuracy $\delta_v$ and $V_K(KD)$ to accuracy $K\delta_v$. Otherwise the probe length can be doubled, ending with a direct solve at $K$. This finite bound requires no assumption on differentiability of the limiting rate. Appendix~\ref{app:certificate} treats separate probe lengths for the slope and variance, relative errors, and numerical precision.

Figure~\ref{fig:limit} shows how much the enclosure can say about the limit itself. For $D\le D_c$ its value is exactly $\nu(p)$. At a differentiability point above $D_c$, \eqref{eq:redun} puts the limiting multiplier in the level set with tolerance $J_0/b$. Equation~\eqref{eq:common} then bounds the limiting standard deviation using only $(1+\sqrt2)\sqrt{\rho/b}$ as the allowance. Above $D_c$, the two-symbol probe gives a valid interval, but its width in Figure~\ref{fig:limit} exceeds the separation among the plotted finite-block curves. At this probe length the bound does not locate the limit to that resolution.

\begin{figure}[htbp]\centering
\includegraphics[width=1\linewidth]{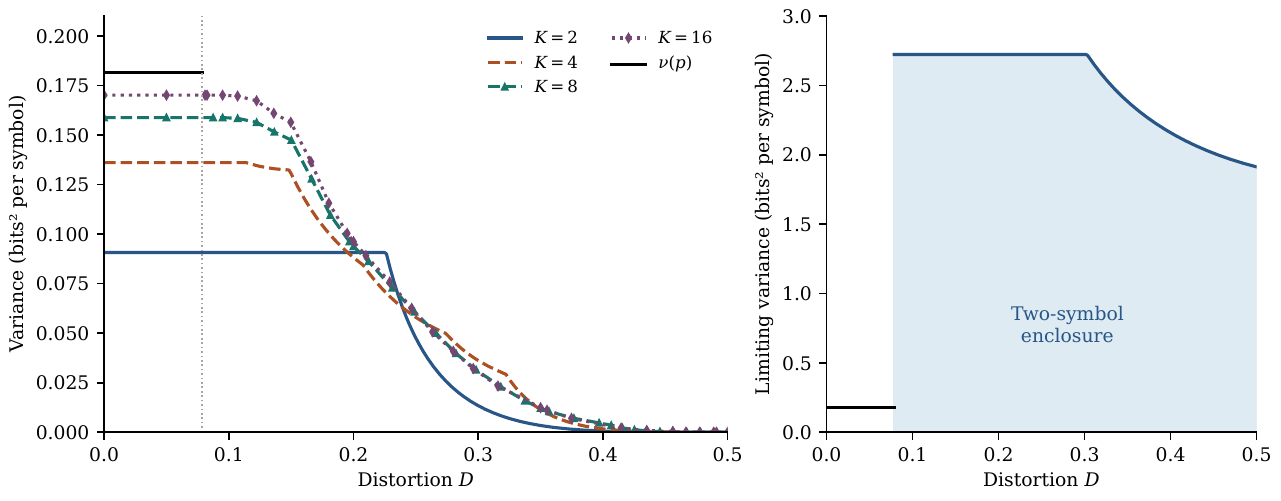}
\caption{Variance versus distortion at $p=0.35$. Left: the short-block formulas give the $K=2,4$ curves. Markers at $K=8,16$ are evaluated points joined by segments, including the exact full-support branch and the zero-slope endpoint. The black segment is the exact limiting plateau. Right: the theoretical enclosure for the limiting variance at interior differentiability points, evaluated numerically with $b=2$ on a separate scale. For a two-symbol probe $\omega_2(\tau)$ is nondecreasing by \eqref{eq:sym2} and \eqref{eq:lowp}, so the endpoints of the multiplier interval suffice.}\label{fig:limit}
\end{figure}
\subsection{A second-order converse for a single path}\label{sec:operational}
For an uninterrupted source path take $K=n$ and set $J_n=\jj_n(X^n,D)$. Its mean is $R_n(nD)$. The rate bracket in \eqref{eq:bracket} implies
\[
 0\le \E J_n-nR(D)\le I(X_1;X_2).
\]
The choice of centering therefore has no effect on the $\sqrt n$ scale. If a length-$n$ code uses $M$ reproduction words, the tilted-information converse \cite[Thm.~7]{kv} bounds its excess-distortion probability by
\begin{equation}
 \epsilon_n\ge\Prb\{J_n\ge\log_2 M+\gamma\}-2^{-\gamma},\qquad\gamma\ge0 .
 \label{eq:markovconverse}
\end{equation}

\begin{proposition}[Low-distortion converse]\label{prop:converse}
Let $0<p<1/2$, $0<D\le D_c$ and $0<\epsilon<1$. For the stationary symmetric Markov chain,
\[
 \liminf_{n\to\infty}\sqrt n\,\bigl(R^*(n,D,\epsilon)-R(D)\bigr)
 \ \ge\ \sqrt{\nu(p)}\,Q^{-1}(\epsilon).
\]
\end{proposition}
\begin{proof}
By \eqref{eq:bin}, $J_n-\E J_n$ is a centered binomial count of $n-1$ independent flips times $\log_2((1-p)/p)$. Thus $(J_n-nR(D))/\sqrt{n\nu(p)}$ converges in distribution to a standard normal, since $\E J_n-nR(D)$ is bounded. If the claimed bound failed, along a subsequence the normalized log size of an optimal code, centered at $nR(D)$ and divided by $\sqrt{n\nu(p)}$, would be at most $Q^{-1}(\epsilon)-2\delta$ for some $\delta>0$. Choose $\gamma=\log_2 n$ in \eqref{eq:markovconverse}. For all sufficiently large $n$ on that subsequence, the probability on the right is at least the probability of exceeding $Q^{-1}(\epsilon)-\delta$ on the normalized scale. Its limit is $Q(Q^{-1}(\epsilon)-\delta)>\epsilon$, while $2^{-\gamma}$ tends to zero. This contradicts the excess-distortion constraint.
\end{proof}

\paragraph{Interpretation.}
For $0<D\le D_c$, the limiting variance $\nu(p)=p(1-p)\log_2^2((1-p)/p)$ is the source varentropy rate. This is the same source fluctuation quantity that enters the lossless theory of Kontoyiannis and Verd\'u \cite{kontoverdu}. Changing $D$ within this interval shifts the score of every path by the same deterministic offset, leaving its fluctuations unchanged. The same identity explains the distortion-independent dispersion of the i.i.d.\ Bernoulli source under Hamming distortion \cite{kv}. For the Markov path, $\sqrt{\nu(p)}\,Q^{-1}(\epsilon)$ is the proved coefficient of the second-order lower bound.

\subsection{Finite-length checks and their limitations}\label{sec:emp}
Figure~\ref{fig:talksix} compares six stationary sources at the same $D=0.10$ \cite{itw}. Independent sources give two exact reference values. For a fair bit $v_K=0$ at every length. For an independent Bernoulli$(0.2)$ bit the two scores differ by two bits, so $v_K=0.2(0.8)\,2^2=0.64$. The correlated sources behave differently even at this single distortion. The biased persistent example rises over these lengths, the slowly mixing example falls overall, and the symmetric persistent example peaks at $K=6$.
\begin{figure}[htbp]\centering
\begin{tikzpicture}
\begin{axis}[blockaxis,width=.62\linewidth,height=5.4cm,xmin=1.8,xmax=12.2,ymin=-.015,ymax=.70,
 xtick={2,4,6,8,10,12},ytick={0,.2,.4,.6},xlabel={Block length $K$},ylabel={Variance per bit $v_K(D)$},
 clip=false,legend style={at={(1.025,.54)},anchor=west,font=\scriptsize,row sep=2.5pt,draw=none,fill=none,inner sep=0pt},
 legend cell align=left]
\addplot[black!60,densely dotted,mark=o,line width=.95pt,mark size=1.8pt] coordinates {(2,0.00000000) (4,0.00000000) (6,0.00000000) (8,0.00000000) (10,0.00000000) (12,0.00000000)};
\addlegendentry{i.i.d. fair $(.5,.5)$}
\addplot[rateorange,solid,mark=square*,line width=.95pt,mark size=1.8pt] coordinates {(2,0.64000000) (4,0.64000000) (6,0.64000000) (8,0.64000000) (10,0.64000000) (12,0.64000000)};
\addlegendentry{i.i.d. biased $(.2,.8)$}
\addplot[rateteal,solid,mark=triangle*,line width=.95pt,mark size=1.8pt] coordinates {(2,0.05707401) (4,0.24683295) (6,0.28725542) (8,0.26879177) (10,0.26483821) (12,0.26217962)};
\addlegendentry{Persistent, fair $(.1,.1)$}
\addplot[rateblue,solid,mark=diamond*,line width=.95pt,mark size=1.8pt] coordinates {(2,0.47744506) (4,0.59443604) (6,0.60820797) (8,0.62173877) (10,0.62968034) (12,0.63477087)};
\addlegendentry{Persistent, biased $(.1,.3)$}
\addplot[violet,dashed,mark=triangle,line width=.95pt,mark size=1.8pt] coordinates {(2,0.34551063) (4,0.19544051) (6,0.14820650) (8,0.12878918) (10,0.12305420) (12,0.12905744)};
\addlegendentry{Slow mixing $(.02,.08)$}
\addplot[brown!85!black,dash dot,mark=+,line width=.95pt,mark size=1.8pt] coordinates {(2,0.32000007) (4,0.32256869) (6,0.33473869) (8,0.33554300) (10,0.33629061) (12,0.33678875)};
\addlegendentry{Anti-persistent $(.8,.8)$}
\end{axis}\end{tikzpicture}
\caption{Tilted-information variance per symbol at $D=0.10$. The pair $(a,b)$ gives the transition probabilities $P(1\mid 0),P(0\mid 1)$. Each marker uses the stationary distribution of all $2^K$ source words. The six sources are those of \cite{itw}.}\label{fig:talksix}
\end{figure}

Figure~\ref{fig:sweep} compares three distortions for each of seven symmetric sources. Table~\ref{tab:nonmonotone} gives two examples at the middle distortion. Both curves rise and then fall, ruling out a general monotonicity argument.

\begin{table}[htbp]\centering\small
\caption{Normalized variance at $D=0.20$ for two symmetric sources, in bits$^2$ per symbol. A dash denotes an unevaluated length.}\label{tab:nonmonotone}
\begin{tabular}{@{}lrrrrrr@{}}\toprule
$p$ & $K=2$ & $8$ & $10$ & $12$ & $18$ & $22$\\\midrule
0.10 & 0.00809 & 0.04067 & 0.08513 & 0.07272 & 0.05960 & 0.06177\\
0.20 & 0.02845 & 0.10852 & 0.11180 & 0.11059 & 0.10849 & --\\\bottomrule
\end{tabular}
\end{table}

\begin{figure}[!htb]\centering
\includegraphics[width=1\linewidth]{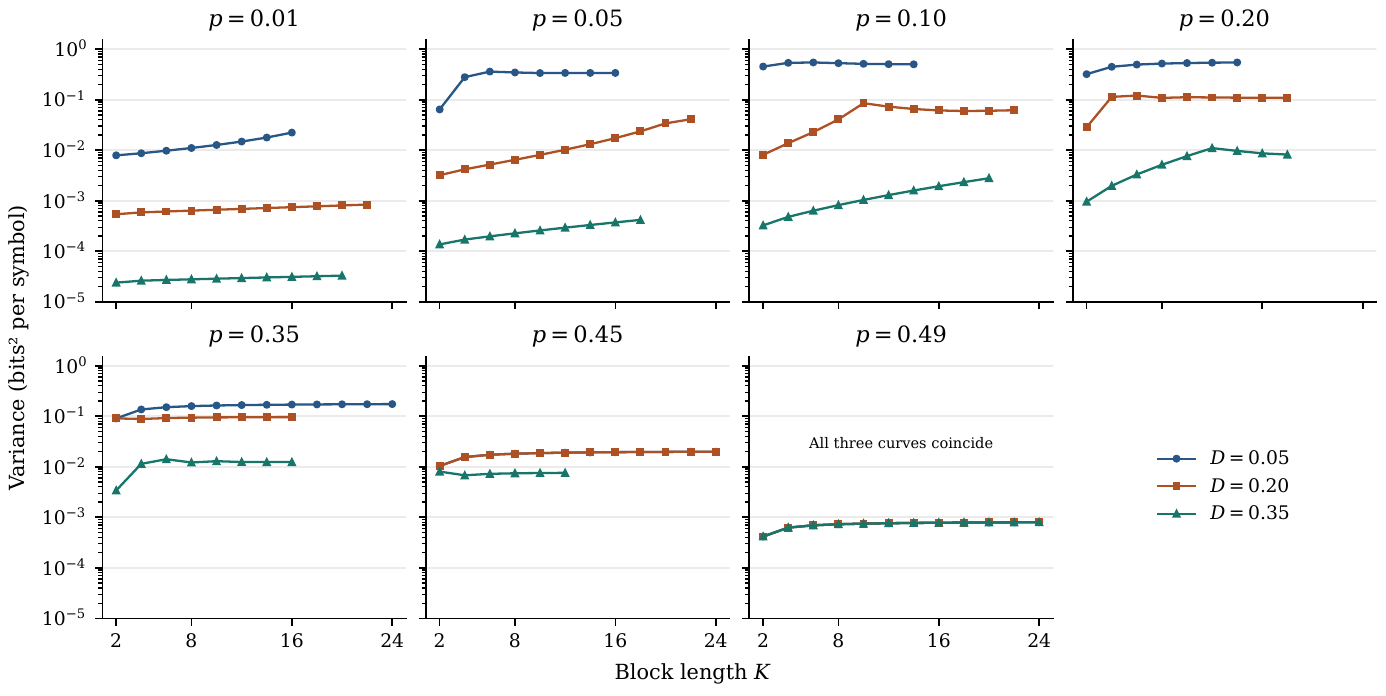}
\caption{Normalized variance at three fixed distortions. Each panel holds the flip probability fixed. Only evaluated block lengths are marked. The common logarithmic scale shows how the magnitude and the dependence on $K$ change across the seven chains.}\label{fig:sweep}
\end{figure}

In the workshop study \cite{itw}, 206 fixed-distortion curves lay outside full support at $K=15$. For 51 of these curves the data through $K=15$ gave inconclusive evidence of stabilization. This occurred more often for slowly mixing sources and for distortions near the zero-rate endpoint. Appendix~\ref{app:data} gives the criteria and subgroup counts, including a case in which apparent stabilization occurs on a temporary reproduction support.
For the slowly mixing source at $K=12$ and $D=0.10$, the Gaussian reference underestimates the standardized 90th percentile ($1.282$ versus $1.708$) and overestimates the 99th ($2.326$ versus $2.143$). Across the five sources in Table~\ref{tab:quant}, twelve of the fifteen quantiles lie within $15\%$ of their Gaussian references. The fair independent source has zero variance, so its standardized quantiles are undefined. Figure~\ref{fig:hist} shows how a change in support alters the distribution of tilted information.

\begin{table}[htbp]\centering\small
\caption{Standardized upper quantiles of $J=\jj_{12}(X,D)$ at $D=0.10$ \cite{itw}. Each entry for a binary source is $(\xi_\alpha-\E J)/\sd J$, where $\xi_\alpha=\inf\{z:\Prb(J\le z)\ge\alpha\}$ uses the probabilities of all $4096$ input words.}\label{tab:quant}
\begin{tabular}{@{}lrrrr@{}}\toprule
Source & $(a,b)$ & 90th & 95th & 99th\\\midrule
Standard normal & & 1.282 & 1.645 & 2.326\\
Biased i.i.d. & $(0.20,0.80)$ & 1.155 & 1.876 & 2.598\\
Persistent, symmetric & $(0.10,0.10)$ & 1.498 & 1.768 & 2.333\\
Persistent, biased & $(0.10,0.30)$ & 1.341 & 1.669 & 2.134\\
Slow mixing & $(0.02,0.08)$ & 1.708 & 2.006 & 2.143\\
Anti-persistent & $(0.80,0.80)$ & 1.233 & 1.836 & 2.235\\\bottomrule
\end{tabular}
\end{table}

\begin{figure}[htbp]\centering
\includegraphics[width=.95\linewidth]{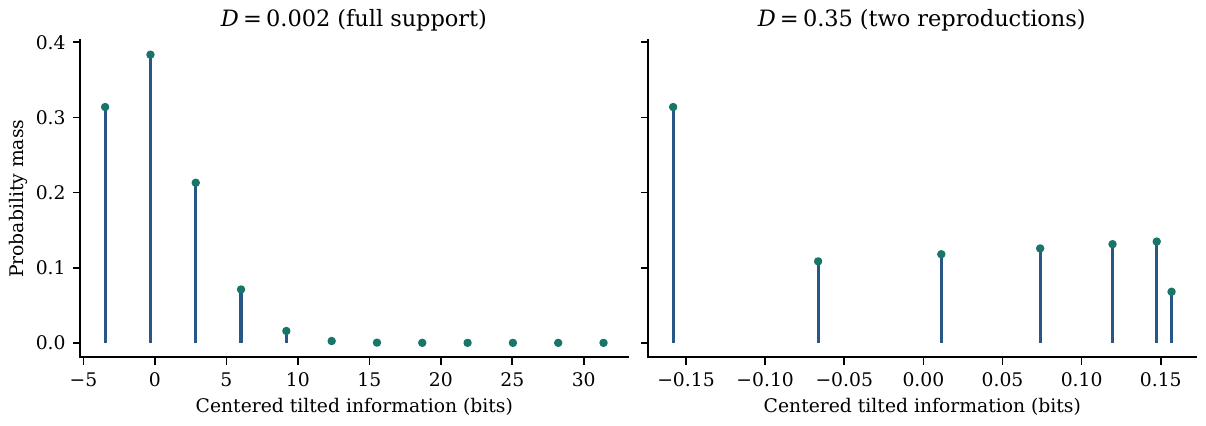}
\caption{Distribution of centered block tilted information for $K=12$, $p=0.1$. At $D=0.002$ each mass corresponds to a binomial flip count. At $D=0.35$ only the two constant reproductions are active; the masses group by distance to the nearer constant. Both distributions use the Markov probabilities of all 4096 inputs. The horizontal scales differ.}\label{fig:hist}
\end{figure}
\section{What remains for operational dispersion}\label{sec:conclusion}
For a symmetric persistent Markov source at $0<D\le D_c$, the central limit theorem for path tilted information gives a second-order converse. At higher interior distortions where $R$ is differentiable, extending the converse requires a distributional result for the path tilted information $J_n$. One route would be to approximate its centered value by an additive functional of the chain with error $o_{\Prb}(\sqrt n)$ and establish a central limit theorem with the variance in Theorem~\ref{thm:common} \cite{gordin,bolthausen}. In both regimes the achievability question concerns the probability of covering a source path.

Draw a source path $X^n$ from the stationary Markov law and independently form a random codebook of $M$ independent words with law $q_n$. For a realized path $x$, let $B_D(x)=\{y:d(x,y)\le nD\}$. The probability that every codeword misses this ball is $(1-q_n(B_D(x)))^M$ \cite[Thms.~9--10]{kv}. A sufficient step toward matching the converse is to find $q_n$ for which
\[
 -\log_2 q_n(B_D(X^n))\le J_n+o_{\Prb}(\sqrt n),
\]
where the remainder divided by $\sqrt n$ converges to zero in probability. The direct coding bounds of Kontoyiannis \cite{konto00} and Kontoyiannis and Zhang \cite{kz02} relate compression to this distortion-ball probability.

Even at $0<D\le D_c$, the explicit optimal reproduction law $q_D$ in \eqref{eq:inverse} leaves this probability to estimate. In the backward channel $X^n=Y^n\oplus E^n$, the Bernoulli$(D)$ errors are independent of $Y^n$. In the covering problem the center $x$ is fixed and the codeword is drawn from $q_D$. Words at the same distance from $x$ can have different probabilities under $q_D$ because their symbols occur in different orders. To match the converse by this route, one must lower-bound the mass of the ball around a typical source path sharply enough to obtain the $o_{\Prb}(\sqrt n)$ remainder above.

\paragraph{Acknowledgment.}
The author used Claude Code, ChatGPT, Gemini, and Grok to help draft and revise the manuscript, develop derivations, check numerical results, and prepare figures. The human author accepts full responsibility for its contents.

\appendix

\section{Finite-block calculations and numerical methods}\label{app:short}
\subsection{Derivations for three and four symbols}\label{app:allocation}
\paragraph{Three-symbol classes.}
A complementary pair of outputs contributes an average weight $\alpha_u=(1+u^3)/2$ to \eqref{eq:Z} if it contains the input $x$, and $\beta_u=(u+u^2)/2$ otherwise. The one-flip class consists of two complementary pairs. If $w_i$ is the reproduction mass of class $i$ ($i=0,1,2$ changes), then
\[
 Z_i=\beta_u+c_iw_i,\qquad (c_0,c_1,c_2)=(\alpha_u-\beta_u)(1,1/2,1).
\]
Maximizing $\sum_i s_i\ln Z_i$ with $\sum_iw_i=1$ gives $w_i=[s_i/\eta-\beta_u/c_i]_+$, with $\eta$ fixed by normalization. Because $s_0>s_1/2>s_2$, the active classes form a prefix in the order $0,1,2$. Moreover $\beta_u/(\alpha_u-\beta_u)=u/(1-u)^2$ increases from zero to infinity. To see that a discarded class stays inactive, put $c_i=(\alpha_u-\beta_u)h_i$. For an active prefix $I$, the inclusion threshold on $s_i h_i$ is $\eta \beta_u/(\alpha_u-\beta_u)=(\sum_{j\in I}s_j)/((\alpha_u-\beta_u)/\beta_u+\sum_{j\in I}1/h_j)$. It increases with $u$ and is continuous at each change.

For the first support threshold, encode bits as spins. Adjacent source spins have correlation $\rho$ and the two end spins have correlation $\rho^2$. Independent backward crossover $D$ multiplies each reproduction pair correlation by $(1-2D)^2$. Complement symmetry makes all odd spin moments vanish. Expanding the probability of either alternating output in these moments gives $[1-(2\rho-\rho^2)/(1-2D)^2]/8$. Its zero gives $D_{\rm FS}(3)$ in \eqref{eq:d3a}.

When $w_2=0$, the two active equations give $\eta=S_{01}(\alpha_u-\beta_u)/(\alpha_u+2\beta_u)$. Hence $Z_0=s_0(\alpha_u+2\beta_u)/S_{01}$, $Z_1=s_1(\alpha_u+2\beta_u)/(2S_{01})$, and $Z_2=\beta_u$. Differentiating the optimized sum $\sum_i s_i\ln Z_i$ with respect to $u$ and multiplying by $u$ gives \eqref{eq:3mid}. The terms from differentiating the optimal masses sum to zero by stationarity on the active classes. The one-flip class reaches zero when $s_1/(2\beta_u)=s_0/\alpha_u$, which gives $u/(1-u+u^2)=p/(1-p)$ after cancellation of $1+u$. Cross-multiplication gives $u/(1+u^2)=p$. This yields $u_2$ in \eqref{eq:d3a}. Computing the distortion with only the two constants active gives \eqref{eq:d3b}.

\paragraph{Four-symbol parity groups.}
Use the quantities $m_g,L_g,k_g$ of \eqref{eq:fourz}. In each group $g=E,O$, order $s_{g,i}$ as in Table~\ref{tab:four}. For $0<p<1/2$, the class probabilities strictly decrease in the indicated order because $1-p>p$. For $0<u<1$, $L_g$ is a convex combination of $\chi_u>0$ and $\beta_u>0$, and $k_E=(1-u^2)^2/2>0$, $k_O=k_E/2$. The active derivatives satisfy $s_{g,i} k_g/Z_{g,i}=\eta_g$, and an inactive derivative is at most $\eta_g$. For any initial segment of the ordered classes, the sum of its unclipped masses cannot exceed $m_g$, with equality for the set of active classes. This proves the maximum formula. At zero group mass the maximand is $k_g/L_g$ times the average of the first $j$ probabilities. The largest average is the first probability.

Transferring mass to the even group contributes $\eta_E-\eta_O$ plus the derivative of the common terms $L_E,L_O$, giving \eqref{eq:G}. Concavity in $\theta$ gives the stated endpoint rules. Differentiating at the optimum with the masses held fixed gives \eqref{eq:fourD}. The uniqueness of the finite interior multiplier is shown in Appendix~\ref{app:certificate}.
\subsection{Numerical methods and stabilization criteria}\label{app:data}
\paragraph{Accuracy of the finite-block calculations.}
For the small-block checks, the reduced solutions and an optimization over the complete reproduction simplex were compared at five flip probabilities and nine multipliers for each of $K=2,3,4$. Per-symbol distortions agreed to within $5\times10^{-9}$ and normalized variances to within $1.5\times10^{-7}$.

For the additional points in Figure~\ref{fig:limit}, a tensor optimization at each slope yields both the horizontal coordinate $D$ and the normalized variance. Points computed iteratively were retained when the largest positive column violation was at most $10^{-5}$ in floating-point arithmetic. The analytic two-reproduction solution was used only after testing every column. For Figure~\ref{fig:sweep}, 75 of 204 entries used \eqref{eq:lowp}, 75 used the verified two-constant solution, and 54 used numerical optimization. Across the 204 entries, the maximum positive column violation was below $1.5\times10^{-7}$ and the maximum per-symbol distortion discrepancy was below $9.0\times10^{-12}$ in floating-point arithmetic. At $D=0.05$, the results for $p=0.05$ and its alternating-mask counterpart $0.95$ agreed to six decimal places at $K=8,10$ ($0.34565,0.33630$).

\paragraph{Classification at length fifteen.}
The workshop study \cite{itw} covered 30 sources with $-0.90\le\rho\le0.90$ and $0.10\le\pi_1\le0.75$. Each of its 307 fixed-distortion curves used the full stationary distribution through $K=15$. We classified support at $K=15$ using the smallest entry of the inverse-channel candidate in \eqref{eq:inverse}. This entry was below $-10^{-10}$ for 206 curves, above $10^{-10}$ for 56, and within the tolerance of zero for 45.

\paragraph{Evidence of stabilization.}
For each curve, the study fitted linear and quadratic functions of $1/K$ over three ranges of block lengths, $8$--$15$, $10$--$15$, and $12$--$15$. The six resulting intercepts estimate the limiting variance. Table~\ref{tab:stability} gives the criteria used to judge their agreement and the spread of the last three computed values. The study withheld estimates when $K$ fell below the $1\%$ worst-start total-variation mixing time or a fitted intercept was materially negative.

\begin{table}[htbp]\centering\small
\caption{Empirical stabilization criteria for fixed-distortion variance curves through $K=15$. The last row applies to negatively persistent sources when both intercepts can be estimated.}
\label{tab:stability}
\begin{tabularx}{\linewidth}{@{}Xcc@{}}\toprule
Comparison & Strong & Moderate\\\midrule
Range of last three values relative to their median & $<1\%$ & $<5\%$\\
Agreement among six extrapolated intercepts & within $2\%$ & within $10\%$\\
Agreement of matched even- and odd-length intercepts & within $2\%$ & within $10\%$\\\bottomrule
\end{tabularx}
\par\smallskip\raggedright For values below $0.01$, the strong tail-spread criterion uses an absolute range below $0.001$.
\end{table}

A nearly stationary curve can also result from a temporary reproduction support. For $p=0.01$, $D=0.20$, the optimal support consists of the two constants at the evaluated even lengths through $K=22$. At $K=110$, the word with a transition halfway through the block gives $c(y)=1.0003$ for the same candidate, violating \eqref{eq:kkt}. At $K=140$ even the nearer constant gives mean distortion $0.20032$ per symbol, above the target. These two checks distinguish loss of optimality from loss of feasibility.

Of the 206 curves confirmed outside full support at $K=15$, 99 met the strong criterion, 56 met only the moderate criterion, and 51 gave inconclusive evidence of stabilization. By mixing time, the inconclusive fraction was $5\%$ among the 85 curves in the 1--5 step group and $71\%$ among the 28 curves above 15 steps. By relative distortion, it was $6\%$ among the 69 curves with $D/D_0<0.25$ and $46\%$ among the 54 curves with $D/D_0\ge0.70$.

For the biased chain $(a,b)=(0.1,0.3)$, the last three values at $D=0.125$ were $0.52296$, $0.52365$, and $0.52427$ bits$^2$ per symbol. The same chain has varentropy rate $1.1930$ on the full-support sequence in Figure~\ref{fig:varent}. At $D=0.2125$ the last three values varied by only $0.54\%$, while extrapolated intercepts disagreed by $26\%$. Agreement of successive variances can thus be a poor estimate of the error at finite $K$.

\section{Proofs and bounds for the variance limit}
\subsection{Comparing block scores across a cut}\label{app:split}
In the proofs of the variance bounds we assume $0<p<1/2$ and measure information in nats. We prove \eqref{eq:common} directly. Let $P$ be the law on two adjacent pieces of lengths $n,k$ and $Q$ the product of their Markov marginals. Write $Z_P,Z_Q$ for the optimal normalizers at the same $\tau$. The normalizer for $Q$ factors. The product of the individual optimal priors satisfies \eqref{eq:kkt} because each column test factors. Put $\Delta=\ln(Z_Q/Z_P)$. For any trial prior $q$, taking its weighted average of the column tests at $P$ gives $\E_P[Z_q/Z_P]\le1$. Interchanging $P$ and $Q$ gives the other comparison. Since $P/Q$ is either $2p$ or $2(1-p)$ according to the boundary flip,
\[
 \E_P e^\Delta\le1,\qquad \E_P e^{-\Delta}\le2(1-p)\E_Q e^{-\Delta}\le2(1-p).
\]
From $x^2\le e^x+e^{-x}-2$ we get $\E_P\Delta^2\le1-2p=\rho$. This proves the second-moment bound in Lemma~\ref{lem:cut} from the column conditions and the density ratio at the cut.

Complement symmetry makes a block score a function of its internal flips alone. The scores of the two pieces are independent under $P$ as well as $Q$. The difference of their sum and the joint score has $L^2(P)$ norm at most $B=\sqrt\rho$. Standard deviations differ by at most this norm. With $\sigma_n=\sqrt n\,\omega_n(\tau)$ we have
\begin{equation}
 \left|\sigma_{n+k}-\sqrt{\sigma_n^2+\sigma_k^2}\right|\le B .   \label{eq:cut}
\end{equation}
Consequently $|\omega_{2n}-\omega_n|\le B/\sqrt{2n}$. Along the sequence $n,2n,4n,\ldots$ the differences sum to at most $(1+\sqrt2)B/\sqrt n$.

Let $\widetilde\omega_n=\lim_{j\to\infty}\omega_{2^jn}$ at the chosen multiplier. Joining $t$ pieces of length $2^jn$ in \eqref{eq:cut} accumulates at most $(t-1)B$ error in the total standard deviation, since $x\mapsto\sqrt{x^2+c^2}$ is 1-Lipschitz for $x\ge0$. After normalization,
\[
 |\omega_{2^j(tn)}-\omega_{2^jn}|\le (t-1)B/\sqrt{t2^jn}.
\]
Letting $j$ grow with $t,n$ fixed gives $\widetilde\omega_{tn}=\widetilde\omega_n$. Taking $n=1$ identifies every limit with $\widetilde\omega_1$. The sum over doublings thus proves \eqref{eq:common} for every starting length, uniformly in $\tau$.

For $\tau>0$, compactness and uniqueness give continuity of the finite optimum. For every prior, $e^{-n\tau}\le Z_{q,\tau}(x)\le1$, so $0\le\ell_{n,\tau}(x)\le n\tau$. At $\tau=0$ the score is identically zero. For fixed $n$, $\omega_n(\tau)\le\sqrt n\,\tau$ gives continuity from the right at zero. Equation~\eqref{eq:common} gives uniform convergence of these continuous functions, hence continuity of $\omega_\infty$ and $\omega_\infty(0)=0$.

The mean comparison in \eqref{eq:redun} follows from $F_n+F_k-J_0\le F_{n+k}\le F_n+F_k$, using the entropy and information bounds for a single cut as in \eqref{eq:bracket} and taking large multiples of $n$.

\subsection{Positivity of the limiting variance}\label{app:positive}
We prove the lower bound \eqref{eq:lower}.
Two bounds on the gain $G_{i,j}=F_i(\tau)+F_j(\tau)-F_{i+j}(\tau)$ from joining blocks are enough. Use binary spins $\pm1$ and denote the optimum correlation at the last site by $r_k^{\rm end}=\E[X_k Y_k]$. Write $b_k=r_k^{\rm end}$ in this proof. Reversibility gives the same correlation at the first site. Complement symmetry gives $\E[Y_k\mid X_k]=b_k X_k$.

\paragraph{Lower bound from the boundary reproductions.}
Apply the optimal $i$- and $j$-block channels independently conditional on their respective source words on the true chain. The distortions add. The joint mutual information is the sum of the short-block informations minus the mutual information of their reproduction words, since the conditional output entropies add. The two reproduction spins adjacent to the cut have correlation $\rho b_i b_j$. Conditional on the two source spins there, the two sides are independent by the Markov property. Their marginals are fair. For fair spins of correlation $c$, their mutual information in nats is $\ln2-h((1-c)/2)\ge c^2/2$, where $h=(\ln2)\hx$. Data processing therefore gives
\begin{equation}
 G_{i,j}\ge\tfrac12\rho^2 b_i^2 b_j^2.                             \label{eq:gainlo}
\end{equation}

\paragraph{Upper bound from the boundary flip.}
Change only the source flip probability at cut $i$ of a $k$-block from $p$ to $1/2$. Under this new law the two source pieces are independent and the optimal cost is $F_i+F_{k-i}$. The original prior is a feasible trial prior for the new law. Thus $G_{i,k-i}$ is at most $(1/2-p)(\E[\ell_{k,\tau}\mid B_i=1]-\E[\ell_{k,\tau}\mid B_i=0])$ with $B_i$ the original flip. Since the flips are independent and $T=\sum_{i=1}^{k-1}B_i$ has variance $(k-1)p(1-p)$, summation and Cauchy--Schwarz give
\begin{equation}
 \sum_{i=1}^{k-1}G_{i,k-i}
 \le\frac{1/2-p}{p(1-p)}\Cov(\ell_{k,\tau},T)
 \le\frac{1/2-p}{\sqrt{p(1-p)}}\sqrt{k(k-1)}\,\omega_k(\tau).             \label{eq:gainhi}
\end{equation}

\paragraph{Combining the gain bounds.}
The uniform reproduction prior gives $F_N/N\le-\ln[(1+e^{-\tau})/2]$. Removing the last source and reproduction spins from an optimal channel gives $F_k-F_{k-1}\ge \tau(1-b_k)/2$ (data processing plus the omitted distortion), with $F_0=0$. Summing yields
\[
 \frac1N\sum_{k=1}^N b_k\ge L(\tau):=2\ln\cosh(\tau/2)/\tau>0 .
\]
Sum \eqref{eq:gainlo} over $1\le i,j\le N$. Its right side is at least $\rho^2N^2L(\tau)^4/2$ by Jensen's inequality. These cuts occur among all cuts for lengths $2,\dots,2N$, and every gain is nonnegative. Sum \eqref{eq:gainhi} over these lengths, divide by $N^2$, and use $\omega_k(\tau)\to \omega_\infty(\tau)$. Since $N^{-2}\sum_{k=2}^{2N}\sqrt{k(k-1)}\to2$ and $1/2-p=\rho/2$, this proves \eqref{eq:lower}. For finite $K\ge2$, the variance at an interior operating point is positive as well. Constant $Z$ would force a uniform $q$ by invertibility of the Hamming kernel. Then every column is active, forcing uniform $P_K$ by \eqref{eq:kkt}, which contradicts $p<1/2$.

\paragraph{A lower bound at fixed distortion.}
For a fixed reproduction word $y$, reveal a fair stationary source one spin at a time and form the Doob martingale for $d(X^K,y)$. The two current states give probabilities of a future mismatch at lag $h$ differing by $|P^h(1-y_{i+h}\mid0)-P^h(1-y_{i+h}\mid1)|=\rho^h$. For a binary indicator this is the total-variation distance, with the convention $d_{\rm TV}=\tfrac12\|P-Q\|_1$. Summing from $h=0$ to $K-i$ bounds the conditional range of the $i$th increment by $c=1/(1-\rho)$. Hoeffding's lemma applied to each increment gives
\[
 \ln\E_{P_K}\exp[-t(d(X^K,y)-K/2)]\le Kt^2c^2/8 .
\]
For any test channel apply the entropy variational inequality to $P_{X\mid Y=y}$ relative to $P_K$, average over $y$ and maximize in $t\ge0$. This gives $I(X;Y)/K\ge2(1-\rho)^2(1/2-D)^2$ in nats. A supporting line for the rate, evaluated at $1/2$, now gives $\tau_K\ge \tau_0=2(1-\rho)^2(1/2-D)$. The same bound holds at a limiting maximizing slope. As $L(\tau)$ increases, \eqref{eq:common} and \eqref{eq:lower} imply
\[
 (\ln2)\sqrt{v_K(D)}
 \ge\frac{\rho\sqrt{p(1-p)}}2 L(\tau_0)^4
       -(1+\sqrt2)\sqrt{\rho/K}.
\]
The lower bound degenerates as $D$ approaches $1/2$.

\subsection{Finite enclosures and a stopping rule}\label{app:certificate}
The two uses of the probe in \eqref{eq:cert} can be separated. Let $b\le j\le K$: use $g_j$ and $I_j$ to locate the slope and $\omega_b$ to estimate the standard deviation. By \eqref{eq:redun}, $|f_j-f_K|\le J_0/j$, so every maximizing $\tau_K$ belongs to $I_j$. If $j$ divides $K$, then $0\le f_j-f_K\le J_0/j$ and the tolerance defining $I_j$ can be halved. At $\tau_K$, the triangle inequality and \eqref{eq:common} give the same $\Delta_{K,b}$ as before. For a direct solve at $b=j=K$ the actual maximizing set suffices and the interblock allowance is zero.

\paragraph{Certificate in exact arithmetic.}
For a trial vector $\hat q\ge0$ with $\sum_y\hat q(y)=1$, form $\hat Z=A_u\hat q$ and $\hat\ell=-\ln\hat Z$. The largest column value is
\[
 C_{\max}=\max_y\sum_x P_b(x)\frac{e^{-\tau d(x,y)}}{\hat Z(x)}\, .
\]
Write $\ell^*=-\ln Z^*$ and $\Delta=\hat\ell-\ell^*=\ln(Z^*/\hat Z)$. Averaging the optimal column inequalities against $\hat q$ gives $\E e^{-\Delta}\le1$. Averaging the trial column values against the optimal prior gives $\E e^\Delta\le C_{\max}$, and averaging them against $\hat q$ gives one, so $C_{\max}\ge1$. Jensen's inequality bounds the excess mean by $\ln C_{\max}$. Also $\Delta^2\le e^\Delta+e^{-\Delta}-2$ implies $\E\Delta^2\le C_{\max}-1$. Both scores use the same random block, so the reverse triangle inequality for centered variables gives $|\sd(\hat\ell)-\sd(\ell^*)|\le\sd(\Delta)\le\sqrt{\E\Delta^2}$. Thus
\[
 0\le\E\hat\ell-F_b(\tau)\le\ln C_{\max},\qquad
 \left|\sd(\hat\ell)/\sqrt b-\omega_b(\tau)\right|\le\sqrt{(C_{\max}-1)/b}.
\]
The maximum runs over all columns, including unused reproduction words.

\paragraph{Finite precision and the multiplier grid.}
A certified computation may start with nonnegative finite-precision weights $w_y$ of positive total mass and define $\hat q(y)=w_y/\sum_z w_z$ exactly. Use outward-rounded intervals for this quotient, the normalizers, and all subsequent operations, with positive lower bounds on the normalizers used as divisors. An upper bound $\overline C_b$ on all column values then bounds the solver error in normalized standard deviation by $\sqrt{(\overline C_b-1)/b}$. For the objective at probe length $j$, the corresponding error is $(\ln\overline C_j)/j$ in nats per symbol. Include interval evaluation error in the total per-symbol objective allowance $\delta_f$ and in the computed standard-deviation endpoints.

Choose a multiplier grid covering $[0,(\ln2)/D]$ with covering radius at most $h$. Each candidate channel contributes a penalized cost per symbol with slope in $[0,1]$, so $g_j=f_j-\tau D$ is 1-Lipschitz. For any $\tau\in I_j$, a node within $h$ has computed objective at most $2J_0/j+h+2\delta_f$ below the computed grid maximum. Here $h$ bounds the drop to the nearby node, and $\delta_f$ allows for error at each of the two nodes being compared. Retaining nodes within this tolerance covers $I_j$.

For $\tau_1\le\tau_2\le\tau_1+h$, let $q_i,Z_i$ denote the two optimizers and normalizers for length $b$. Testing $q_2$ in the column inequalities at $\tau_1$ gives $\E(Z_2/Z_1)\le1$ because the kernel decreases with the multiplier. In the other direction, $e^{-\tau_1 d}\le e^{bh}e^{-\tau_2 d}$, so testing $q_1$ at $\tau_2$ gives $\E(Z_1/Z_2)\le e^{bh}$. The same exponential inequality and reverse triangle inequality as above bound the change in $\omega_b$ by $\sqrt{(e^{bh}-1)/b}$. In \eqref{eq:cert}, enlarge the cut allowance $\Delta_{K,b}$ by this grid term and the solver error, using an upper bound valid at all retained nodes or widening each node's endpoints separately. Apply the lower truncation, square the endpoints, and divide by $(\ln 2)^2$.

\paragraph{Stopping at a prescribed accuracy.}
For an additive tolerance $\delta_v$, double $b=j$ until the interval in \eqref{eq:cert} has width at most $2\delta_v$, or solve directly at $K$. At finite $K$ the maximizing multiplier is unique at an interior distortion for a symmetric source. Otherwise $F_K$ would be affine between two maximizers. The concave fixed-prior objective at their midpoint lies above this affine segment and touches it at the midpoint, so it must agree throughout. Its second derivative is minus the expected conditional variance of the Hamming distance. At the positive midpoint, complement symmetry gives positive reproduction mass to both $y$ and $\bar y$ for any active $y$. On input $x=y$, which has positive source probability, both outputs have positive conditional probability and their distances are $0$ and $K$. The expected conditional variance is strictly positive, a contradiction. By rate--distortion duality, the unique maximizing multiplier gives the unique supporting slope of $r_K$ at the chosen interior distortion. Conversely, at a fixed $\tau>0$ the optimal marginal is unique by Section~\ref{sec:full}. The envelope derivative $f'_K(\tau)=\E d(X,Y)/K$ therefore gives a unique attained distortion.

If $R$ is differentiable at every interior distortion, the cutoff for additive precision can be taken uniformly in $K,D$ at fixed $p$. Use the exact formula below $D_c$ and the compact multiplier interval $[0,\ln((1-D_c)/D_c)]$ above it. At $D=1/2$ the uniform prior gives $f(\tau)\le-\ln[(1+e^{-\tau})/2]<\tau/2$ for $\tau>0$, so zero is the unique maximizer there. By compactness and \eqref{eq:redun}, two separated near-maximizers along any subsequence would yield two limiting maximizers at the same distortion. Hence the diameters of $I_j(D)$ shrink uniformly. Continuity and \eqref{eq:common} then force the enclosure widths to vanish uniformly for $K\ge j$. This guarantees additive error $\delta_v K$ for the total variance. Fixed absolute accuracy on the total requires per-symbol precision of order $1/K$.

For $K=1$ the source is a fair bit and $v_1(D)=0$ exactly. For $K\ge2$, the midpoint of the bit-unit interval $[L,U]$ attains relative error at most $\eta$ once $L>0$ and $U-L\le2\eta L$. At a fixed interior distortion where $R$ is differentiable, the enclosure width tends to zero and Theorem~\ref{thm:common} gives a positive limiting variance. The lower endpoint is then bounded away from zero for sufficiently large probe and block lengths, so the relative-error rule has a finite cutoff.
\end{document}